\documentclass[12pt]{article}
\usepackage{amsfonts}

\usepackage{graphicx}
\usepackage{amsmath}
\usepackage{hyperref}

\usepackage{color}

\usepackage{lipsum} 

\newtheorem{theorem}{Theorem}[section]

\newtheorem{corollary}[theorem]{Corollary}

\newtheorem{definition}[theorem]{Definition}

\newtheorem{lemma}[theorem]{Lemma}

\newtheorem{notation}[theorem]{Notation}

\newtheorem{proposition}[theorem]{Proposition}
\newtheorem{remark}[theorem]{Remark}

\newenvironment{proof}[1][Proof]{\textbf{#1.} }{\ \rule{0.5em}{0.5em}}

\def \<{\langle}
\def \>{\rangle}
\def \l{\lambda}

\def \kf{\frak k}
\def \p{\partial}
\def \beq{\begin{equation}}
\def \eeq{\end{equation}}
\def \eref{\eqref}

\def \n{\nabla}

\def \lrc{\lrcorner}
\def \w{\omega}
\def \A{{\mathcal A}}
\def \B{{\mathcal B}}
\def \C{{\mathcal C}}
\def \D{{\mathcal D}}
\def \E{{\mathcal E}}
\def \F{{\mathcal F}}
\def \G{{\cal G}}
\def \H{{\cal H}}

\def \H^0{{\cal H}^0 or}

\def \L{\Lambda}

\def \R{\mathbb R}
\def \V{{\mathcal V}}

\def \Y{{\cal Y}}

\numberwithin{equation}{section}

\begin{document}

\title{Initial behavior of solutions to the Yang-Mills heat equation.
      }

\author{Nelia Charalambous\thanks{The first author was partially supported by a University of Cyprus Start-Up grant.}
  \and
 Leonard Gross
}

\newcommand{\Addresses}{{
  \bigskip
  \footnotesize

 N.~Charalambous, \textsc{Department of Mathematics and Statistics, University of Cyprus, Nicosia, 1678, Cyprus}\par\nopagebreak
  \textit{E-mail address}: \texttt{nelia@ucy.ac.cy}

  \medskip

  L.~Gross, \textsc{Department of Mathematics, Cornell University,  Ithaca, NY 14853-4201, USA}\par\nopagebreak
  \textit{E-mail address}: \texttt{gross@math.cornell.edu}

}}

\date{\today}

\maketitle

\abstract{We explore  the small-time behavior of  solutions to the Yang-Mills heat equation
with rough initial data.
We consider solutions $A(t)$  with initial value  $A_0\in H_{1/2}(M)$, where $M$ is
 a bounded convex region in $\R^3$ or all of $\R^3$. The behavior, as $t\downarrow 0$,
 of the $L^p(M)$ norms of the  time derivatives
 of $A(t)$ and its curvature $B(t)$ will be determined for $p=2$ and $6$, along with the $H_1(M)$
 norm of these derivatives.
}

\tableofcontents

\section{Introduction} \label{secIntro}

In this article we study the initial behavior of solutions to the Yang-Mills heat equation over
a region $M$ in $\R^3$.  Denote by $K$ a compact connected Lie group with Lie algebra $\kf$.
A $\kf$ valued 1-form over $M$ may be written as
\beq
A = \sum_{j=1}^3 A_j(x) dx^j,                         \label{I1}
\eeq
with coefficients  $A_j(x) \in \kf$.
The curvature of $A$ is the $\kf$ valued 2-form given by $B=dA + A\wedge A $.
The Yang-Mills heat equation  is the weakly parabolic equation for a time dependent
$\kf$ valued 1-form $A(t)$ over $M$  given by
\beq
\frac{\p}{\p t} A(x,t) =  -d_{A(t)}^* B(x,t),                 \label{I2}
\eeq
wherein $d_A^* =d^* +$[the interior product by $ad\, A(t)$], and $B(x,t)$ is the curvature of $A(t)$ at $x$.
We will always take $K$ to be a subgroup of the orthogonal, respectively unitary, group of a finite dimensional
real, respectively complex, inner product space $\V$.

The Yang-Mills heat equation is only weakly parabolic since the second
 order derivative terms on the right side  of \eref{I2} are $-d^*d A$, which are missing
  `half' of the Laplacian on 1-forms $-\Delta=d^*d + d d^*$.
  In \cite{CG1} we  proved the existence and uniqueness of solutions to this equation
  for initial data $A_0$ in $H_1(M)$. In \cite{G70} the existence and uniqueness was proven
  for initial data in $H_{1/2}(M)$. The Sobolev index 1/2 is the critical index for the Yang-Mills
  heat equation in spatial dimension three. We will be concerned with solutions to \eref{I2}
  for which the initial value $A_0$ is in $H_{1/2}(M)$. In this case the curvature
  $B(t)$ can be expected to blow up in the $L^2(M)$ sense as $t\downarrow 0$ since,
  informally, $B(t)$    can be expected
  to converge to its initial value $B_0$
   only in the sense of the negative    Sobolev space $H_{-1/2}(M)$.
  Higher derivatives of  $A(t)$ can be expected to blow up more quickly as $t\downarrow 0$. Our study is motivated by a desire to understand the nature of the singularities of gauge
 covariant derivatives of a solution to the Yang-Mills heat equation as time decreases to zero.
  In this article we will study the $L^p(M)$ behavior   of various gauge covariant derivatives
   of $A(t)$ as $t\downarrow  0$. The values $p=2$ and $p=6$ (and a fortiori all $p$ in between)
   are of sole interest in this paper because only first order Sobolev inequalities  can be
    effectively used in our energy methods.
    Concerning higher values of $p$ see Remark \ref{remptwise}.
     Apriori estimates of first, second and third order spatial covariant
     derivatives have already been used in  our previous work \cite{CG1,CG2, G70} to prove existence
     and uniqueness of solutions to \eref{I2}.

A function $g:M\rightarrow K$ induces
a gauge transformation of a time dependent connection form on $M$ by the definition
\beq
A^g(x,t) = g(x)^{-1} A(x,t) g(x) + g(x)^{-1} dg(x).                    \label{I3}
\eeq
If $A(\cdot, \cdot)$
 is a solution to  the Yang-Mills
heat equation \eref{I2} then so is $A^g(\cdot, \cdot)$,
at least if $g$ satisfies some mild  regularity conditions. It is already clear from this
that the Yang-Mills heat equation does not smooth all initial data, for if
 $A(x,t)$ is a solution with initial value  $A_0(x)$ then $A^g(x,t)$ is the solution
  with initial value $A_0^g(x)$,  and consequently, even if $A(x,t)$ is very smooth,
  the solution $A^g(x,t)$ need be no smoother  than $g^{-1}dg$. Our goal is to show
  that solutions to \eref{I2} are infinitely differentiable in a gauge covariant sense
  for $t >0$, even for rough initial data,
  and to determine the nature of the singularities of the derivatives as $t\downarrow 0$.
  For the class of initial
  data that we are interested in, namely $A_0 \in H_{1/2}(M)$, the formula \eref{I3} suggests that
  the corresponding class of allowed gauge functions should include functions $g \in H_{3/2}(M)$.
  A precise definition of this class, which makes it into a complete topological group, will be
  given in Section   \ref{secstate}. With these initial data, which are in fact an invariant class
  under  these gauge transformations, it can be seen easily from \eref{I3} that a solution
  need not be even once continuously differentiable in the ordinary sense.
  There are in fact solutions that are not in the Sobolev space $W_1(M)$ for any $t >0$.
  We are going to address this by computing only gauge covariant derivatives.
 The $L^p(M)$ norm of such a derivative is fully gauge invariant and therefore descends
 to a function on the quotient space $\C \equiv \{\text{connection forms}\}/ \text{Gauge group}$,
 which is a  space of connections over $M$ as well as a version of the  configuration space
 for the classical Yang-Mills field theory. We will establish bounds on these gauge invariant
 norms by functions
 of the action of the solution $A(\cdot)$, which are also fully gauge invariant and which therefore
 also descend to functions on $\C$. We obtain thereby bounds on the covariant derivatives
 given by inequalities between functions on the quotient space $\C$ itself.
  It will be shown in \cite{G72} that the  space $\C$
 has a natural Riemannian metric on it which makes it into a complete Riemannian manifold.
  Our main results can  be interpreted as analysis
 over this manifold. Remark \ref{remqs} makes this a little more precise.

The technical problem of computing high order derivatives of non-differentiable functions
will be carried  out by gauge transforming a solution to a smooth function, which can be done for a short time, \cite{G70},  and then gauge transforming the derivative back.

  For our choice of the region $M \subset \R^3$ we will take either $M=R^3$ or
 take $M$ to be the closure of a bounded open convex subset of $\R^3$ with smooth boundary.
 Undoubtedly  our methods will apply to other regions also with minor modification
 as well as to other manifolds.
 For example, they can be applied to compact three manifolds without boundary,
  and compact three manifolds with convex boundary.
 But we are going to focus just on regions in $\R^3$, which we believe to be adequate
 for our anticipated applications to quantum field theory. In case $M\ne\R^3$ we must
 impose boundary conditions on $A(t)$ for $t >0$. The two natural boundary conditions
 that we will use are the Neumann-like boundary conditions
 (absolute boundary conditions in the sense of Ray and Singer \cite{RaS})
 or the Dirichlet-like boundary conditions (relative boundary conditions).
 For our anticipated applications to quantum field theory we will also ultimately
 need to use Marini boundary conditions, introduced in \cite{Ma3,Ma4,Ma7,Ma8},
which set the normal component of the curvature to zero on the boundary.
Results for Marini
boundary conditions will be deduced elsewhere from our results for
 Neumann-like boundary conditions.

\section{Statement of Results} \label{secstate}

\subsection{Notation.} \label{secnote}

Throughout this paper $M$ will denote either $\R^3$ or the closure of a bounded open
 set in $\R^3$ with smooth boundary. In the latter case we will always assume that $M$ is convex in the sense that the second fundamental form of $\p M$ is
 everywhere non-negative.

 We consider a product bundle over $M$,  $M\times \V \rightarrow M$,  where $\V$ is a finite dimensional  real or complex vector space with an inner product. Let  $K$  be a compact connected subgroup of the orthogonal, respectively unitary,
  group in $End\ \V$. We denote by $\frak k$ the Lie algebra  of $K$, which
  can be identified with a  real  subspace of $End\ \V$.

Let  $\<\cdot, \cdot\>$  be an  $Ad\ K$ invariant inner product on $\frak k$
 with associated norm $|\xi |_{\frak k}$ for $\xi \in \frak k$. For $\frak k$ valued
  $p$-forms $\w$ and $\phi$  the $L^2$ pairing  is given by
  $(\w, \phi) = \int_M\<\w(x), \phi(x)\>_{\L^p\otimes \frak k} d\, \text{Vol}$
    with induced $L^2$ norm $\|\w \|_2^2 = (\w, \w)$.
   We define the $W_1$ norm of $\w$ by
\beq
\| \w \|_{W_1(M)}^2
 = \int_M |\n \w|_{\R^3\otimes \Lambda^p\otimes\frak k}^2 d\, \text{Vol}\ \
                                        + \| \w \|_2^2,    \label{ymh2}
 \eeq
where  $\n \w$ is constructed from the weak derivatives.
 Define $W_1 = W_1(M) = \{ \w: \|\w\|_{W_1(M)} <\infty\}$. This is the Sobolev space of order one, without boundary conditions.

If $u= \sum_{|I|=r} u_I dx^I$ and $ v =\sum_{|J|=p} v_J dx^J$ are $End \ \V$  valued forms then their wedge product, $u\wedge v = \sum_{I,J} u_I v_J dx^I\wedge dx^J$, is another   $End\ \V$ valued form.   When the appropriate action of $u$ on $v$ is via $ad\; u$  then we will write $[u\wedge v] = \sum_{I,J} [u_I, v_J]dx^I\wedge dx^J$. This will be the case, for example,  when $u$ is an $End\ \V$ valued connection form or its time derivative. If $u$ and $v$ take their values in $\frak k$ then so does $[u\wedge v]$.

The interior product,  $[u\lrc v]$, of an element $u \in \L^p \otimes \kf$ with an element $v \in \L^{p+r} \otimes \kf$ is defined, for $r \ge 0$, by
\beq
\< w, [u\lrc v]\> _{\L^r \otimes \kf} = \< [ u\wedge w ], v\>_{\L^{p+r} \otimes \kf} \ \ \text{ for all}\ \ w \in \L^r\otimes\kf.     \label{C9}
\eeq
If $u$ and $v$ are both in $\Lambda^1 \otimes \frak k$ then \eref{C9} gives
\begin{equation*}
\frak k \ni [u\lrc v] = - [u\cdot v] = -\sum_j [u_j, v_j]
\end{equation*}
in an orthonormal frame for $\L^1$.  In particular $[u\lrc u] =0$.
 Moreover, if $w\in \L^2\otimes \frak k$ then $[w\lrc w] =0$.

In this paper we will be concerned
with a $\kf$-valued 1-form $A$ as in \eref{I1}.
For $\w \in W_1(M;\L^p\otimes \kf) $ define  $d_A \w = d \w + [A \wedge \w]$.
Then  $d_A^* \w = d^*\w + [ A \lrc \w]$. The curvature of $A$ can be represented as
\beq
B = dA +(1/2) [A\wedge A] .                               \label{ymh3}
\eeq

\subsection{Strong solutions and boundary conditions.}

We take the following definition of strong and almost strong solution from \cite{G70}.

\begin{definition}\label{defstrsol} {\rm Let $0 < T \le \infty$.
A {\it strong solution} to the Yang-Mills heat equation over $[0, T) \times M$ is a continuous function
\begin{equation*}
A(\cdot): [0,T) \rightarrow  L^2(M; \Lambda^1 \otimes   \frak k )
\end{equation*}
such that
\begin{align}
a)& \ A(t) \in W_1 \ \text{for all} \ t \in (0,T) \ \text{and} \  A(\cdot): (0,T) \rightarrow W_1 \ \text{is continuous},  \notag \\
b)& \ B(t):= dA(t)+ \frac 12 [A(t)\wedge A(t)]  \in W_1 \  \text{for each}\ \  t\in (0,T),\          \notag     \\
c)& \  \text{the strong $L^2(M)$ derivative $A'(t) \equiv dA(t)/dt $}\
                                      \text{exists on}\ (0,T),   \text{and}       \notag \\
 &\ \ \ \ \ \  A'(\cdot): (0, T) \rightarrow L^2(M) \ \text{is continuous},    \notag  \\
d)& \  A'(t) = - d_{A(t)}^* B(t)\ \ \text{for each}\ t \in(0, T).        \notag
\end{align}
A solution $A(\cdot)$ that satisfies all of the above conditions
 except for $a)$ will be called an {\it almost strong solution}. In this
 case the spatial exterior derivative $dA(t)$, which appears in the
 definition of the curvature,  must  be interpreted in the weak sense.
}
\end{definition}

\begin{definition}\label{defbdyconds} {\rm  If $M\ne \R^3$ then we will impose
boundary conditions on the solutions.
For a strong solution to the Yang-Mills heat equation we will consider two types of boundary conditions:

\noindent
{\it Neumann boundary conditions:}
\begin{align}
&i)\ \ \  A(t)_{norm} =0\ \ \text{for}\ t > 0 \ \text{and}    \label{N1}\\
&ii)\ \ B(t)_{norm}=0\ \   \text{for}\  t >0.         \label{N2}
\end{align}
\noindent
{\it Dirichlet boundary conditions:}
\begin{align}
 &i)\ \ \  A(t)_{tan} =0\ \ \ \text{for}\ t > 0 \ \text{and}    \label{D1}\\
  &ii)\ \  B(t)_{tan} =0 \ \ \ \text{for}\ t >0.                 \label{D2}
 \end{align}
 }
 \end{definition}

In \cite{CG1} we also considered
Marini boundary conditions,
which only  require $ B(t)_{norm} = 0$. Solutions satisfying these boundary conditions
will be derived in a later work from solutions satisfying Neumann boundary conditions.
The regularity theorems of the present paper will carry over to these.
We will not consider them in this paper.

\begin{notation}\label{notsob} {\rm The Sobolev spaces for  $\kf$ valued 1-forms associated  to
 the  preceding boundary conditions are most easily described in terms of the related Laplacian.

 If $M= \R^3$  define
 \beq
-\Delta = d^* d + d d^*,                                 \label{Lap1}
\eeq
where $d$ denotes the closed version of the exterior derivative operator
with $C_c^\infty(\R^3, \L^1\otimes \kf)$ as a core.

If $M \ne \R^3$ then the Neumann and Dirichlet Laplacians are again given by
$\sum_{j=1}^3 \p_j^2$
but
subject to the following boundary conditions.
\begin{align*}
&\w_{norm} =0  \ \ \text{and} \ \ (d\w)_{norm}=0\ \   \text{Neumann conditions}    \\
&\w_{tan} =0\ \ \ \ \text{and}  \ \   (d^*\w)_{\p M} =0  \  \ \  \text{Dirichlet conditions}.
\end{align*}
Alternatively, the Neumann, respectively Dirichlet, Laplacian can be defined by \eref{Lap1},
wherein $d$ is taken to be the maximal, respectively minimal, exterior derivative operator over $M$.
See \cite{CG1} for further discussion of these domains.
In all three cases
the Laplacian is
 a nonnegative,
 self-adjoint operator on the appropriate domain.

 For $0 \le a \le 1$ we define the Sobolev spaces
\beq
H_a = \text{Domain of} \ (-\Delta)^{a/2} \ \text{on} \ L^2(M; \Lambda^1\otimes \mathfrak{k}) \notag
\eeq
with norm
\beq
\|\w\|_{H_a} = \|(1-\Delta)^{a/2} \w\|_{ L^2(M; \Lambda^1\otimes\mathfrak{k})}.  \label{Lap5}
\eeq
In this paper we will only be concerned  with the cases $a= 1/2$ and $a =1$.
But it may be interesting to note that for each number $a \in [0,1]$ the two Sobolev spaces $H_a$, corresponding to Dirichlet or Neumann boundary conditions, are distinct when $1/2 \le a \le 1$ and are identical if $0\le a < 1/2$, by Fujiwara's theorem \cite{Fuj}.

With the preceding  definition of a Sobolev space, we have the following embedding property
\[
\|\w\|_{H_a} \leq c_{a,b} \|\w\|_{H_b} \ \text{whenever} \ 0\leq a \leq b.
\]
The constant $c_{a,b}$ is independent of $M$.
}
\end{notation}

\begin{definition}\label{defgg} {\rm (The gauge group $\G_{3/2}$.) A measurable function
$g:M\rightarrow K \subset End\ \V$ is a bounded function into the linear space
$End\ \V$ and consequently its
weak derivatives are well defined. Following \cite{G70} we will write $g \in W_1(M;K)$
if $\|g - I_\V\|_2 <\infty$ and the derivatives $\p_j g \in L^2(M; End\ \V)$.
The 1-form $g^{-1} dg := \sum_{j=1}^3 g^{-1}(\p_j g)dx^j$ is then an a.e. defined  $\kf$ valued 1-form.
The Sobolev norm $\|g^{-1}dg \|_{H_a}$ is defined as in  \eref{Lap5}.
For an element $g \in W_1(M;K)$
 the restriction $g|\p M$ is well defined almost everywhere on $\p M$ by a Sobolev trace theorem.
The three versions of $\G_{3/2}$ that we will need are given in the following definitions.
\begin{align*}
\G_{3/2}(\R^3)
   = \Big\{g \in W_1(\R^3; K): g^{-1}dg \in H_{1/2}(\R^3;\L^1\otimes \kf) \Big\}, \qquad \qquad \ \ \
\end{align*}
If $M \ne \R^3$ define
\begin{align*}
\G_{3/2}^N(M) &= \Big\{g \in W_1(M; K): g^{-1}dg \in H_{1/2}(M;\L^1\otimes \kf) \Big\}, \\
\G_{3/2}^D(M) & = \Big\{g \in W_1(M; K): g^{-1}dg \in H_{1/2}(M;\L^1\otimes \kf),\  g = I_\V\ \text{on}\ \p M \Big\},
\end{align*}
It should be understood that the two spaces denoted  $H_{1/2}(M;\L^1\otimes \kf)$ are those
determined by Neumann, respectively Dirichlet,  boundary conditions. It was proved in \cite[Theorem 5.3]{G70}
that all three versions of $\G_{3/2}$ are complete topological groups in the metric
$dist(g,h) = \| g^{-1} dg - h^{-1} dh\|_{H_{1/2}} +\| g-h\|_{L^2(M; End\, \V)}$.
}
\end{definition}

\begin{definition}  {\rm  A solution $A(\cdot)$ to the Yang-Mills heat equation is said to
 have  {\it finite action } if
\beq
\rho(t):= (1/2)\int_0^t s^{-1/2} \|B(s) \|_2^2\, ds <\infty        \label{dfa}
\eeq
for some $t>0.$ If  $A(\cdot)$ has finite action then, actually, $\rho(t) < \infty$ for all $t >0$ because
$\|B(s)\|_2^2$ is nonincreasing. See e.g. Lemma \ref{lemFE}.
}
\end{definition}

It was shown in \cite{G70}  that a solution to the Yang-Mills heat equation with initial
 value $A_0 \in H_{1/2}$ will have finite action whenever $\|A_0\|_{H_{1/2}}$ is sufficiently small.
 We summarize some of the results needed from \cite{G70} in the following theorem.

\begin{theorem} \label{thmFA} $($\textup{\cite[Theorem 2.11]{G70}}$)$
  Assume that $A_0\in H_{1/2}(M; \Lambda^1 \otimes   \frak k )$. Then there exists an almost strong solution $A(t)$ to the Yang-Mills heat equation over $[0,\infty)$ with initial value $A_0$.

If $\|A_0\|_{H_{1/2}}$ is sufficiently small then   there exists a gauge function $g_0 \in \mathcal{G}_{3/2}$
 such that the connection $A(t)^{g_0}$ is a  strong solution
 over $[0,\infty)$ with initial value $A_0^{g_0}$.
 It is also smooth over $(0,T)\times M$ for some $T < \infty$.
 The solutions $A(t)$ and $A(t)^{g_0}$ have  the following properties   in this case:

\begin{enumerate}

\item  Both $A(t)$ and $A(t)^{g_0}$ are continuous functions on $[0,\infty)$ into $H_{1/2}(M; \Lambda^1 \otimes   \frak k )$.

\item The curvatures of $A(t)$ and $A(t)^{g_0}$  satisfy \eref{N2} in the Neumann case and \eref{D2} in the Dirichlet case for all $t>0$. The gauge regularized solution $A(t)^{g_0}$ satisfies in addition \eref{N1} in the Neumann case and \eref{D1} in the Dirichlet case for all $t>0$.

\item Both $A(t)$ and $A(t)^{g_0}$ have finite action.
\end{enumerate}
\end{theorem}

\begin{remark}{\rm It is also proved in \cite{G70} that strong solutions with finite action are unique
when $M =\R^3$ and, if $M\ne \R^3$, unique under the boundary conditions
  \eref{N2} in case of Neumann boundary conditions
or \eref{D1} in case of Dirichlet boundary conditions. The smoothness of $A^{g_0}(t)$ on $(0,T)\times M$
may hold for the same fixed $g_0$ for $T= \infty$, but this is still an open question.
}
\end{remark}

\subsection{The Main Theorem.}

We are going to establish bounds on various gauge covariant
derivatives  of a solution to the Yang-Mills heat equation
 in terms of the  action functional $\rho(t)$, defined in \eref{dfa}.
 The class of solutions of interest
 are those for which the initial value $A_0$ has small $H_{1/2}$ norm.
 But $\|A_0\|_{H_{1/2}}$ is not  a gauge invariant function of  $A_0$.
In the next theorem we will show that the  gauge invariant functionals of derivatives of $A(\cdot)$
that are of interest to us are controlled  by the gauge invariant functional $\rho$.
The inequalities that implement
this descend therefore to inequalities on the quotient space \{initial data space\}$/ \G_{3/2}$,
thereby  yielding analysis on the quotient space itself.
See Remark \ref{remqs} for further discussion.

By a {\it standard dominating function} we will mean a function
  $C:[0, \infty) \rightarrow [0,\infty)$ of the form $C(t) =\hat C(t, \rho(t))$ , where
  $\hat C:[0,\infty)^2 \rightarrow [0, \infty)$ is continuous and non-decreasing
   in each variable, $\hat C(0,0) =0$ and $\hat C$ is independent of the solution $A(\cdot)$.

Our main result is the following.

\begin{theorem} \label{MainTh} Assume that $A_0\in H_{1/2}(M; \Lambda^1 \otimes   \frak k )$.
Suppose that $A(\cdot)$ is a strong solution to \eref{I2}
over $[0,\infty)$
with initial value $A_0$ and having finite action. If $\|A_0\|_{H_{1/2}}$ is sufficiently small
then there exists $T>0$ and
standard dominating functions $C_{nj}$  for $\ j=1,\ldots 4$
and  $n=1,2,\dots$,   such that, for $0<t < T$,  the following estimates hold.
\begin{align}
\tag{$\A_n$}
 t^{2n-\frac 12}\|A^{(n)}(t)\|_2^2 \ +
     &\int_0^t  s^{2n- \frac 12} \|B^{(n)}(s) \|_2^2\, ds  \le C_{n1} (t) \ \ \ \ \     \label{An1}    \\
\tag{$\B_n$}
t^{(2n- \frac 12)}\|B^{(n-1)}(t)\|_6^2 \ +
    & \int_0^t  s^{2n- \frac 12} \|A^{(n)}(s) \|_6^2 \,ds  \le C_{n2}(t)      \label{BnL6}      \\
\tag{$\C_n$}
t^{2n+\frac 12} \| B^{(n)} ( t)\|_2^2 \ +
     & \int_0^t  s^{2n+\frac 12} \| A^{(n+1)} (s) \|_2^2 \,ds  \le  C_{n3} (t)  \label{Bn1}        \\
\tag{$\D_n$}
t^{2n+\frac 12} \| A^{(n)} (t)\|_6^2  \ +
     & \int_0^t  s^{2n+\frac 12} \| B^{(n)}(s)\|_6^2 \,ds \le C_{n4}(t).           \label{AnL6}
\end{align}
Moreover \eref{Bn1} holds for $n =0$.
\end{theorem}

\begin{notation}\label{ginvsob}
{\rm The gauge invariant version of the Sobolev 1-norm \eref{ymh2}
 is defined by
\begin{align*}
\| A^{(n)}(t)\|_{H_1^A}^2  &= \sum_{j=1}^3 \int_M|\p_j^{A(t)} A^{(n)}(t)|^2 dx + \|A^{(n)}(t)\|_2^2,\ \ n \ge 1,\\
\| B^{(n)}(t)\|_{H_1^A}^2  &= \sum_{j=1}^3 \int_M|\p_j^{A(t)} B^{(n)}(t)|^2 dx + \|B^{(n)}(t)\|_2^2, \ \ n\ge 0,
\end{align*}
where
\[
\p_j^{A(t)}\, \w  =\p_j\,  \w + ad\, A_j(t)\, \w
\]
for a $\kf$ valued p-form $\w$.
}
\end{notation}

\begin{corollary}\label{Maincor}Under the hypotheses of Theorem \ref{MainTh} there exists $T>0$
 and standard dominating functions $C_{nj}$   for $j=5,6$  and $n=1,2,...$  such that, for  $0 < t <T$,  the following estimates hold.
\begin{align*}
\tag{$\E_n$} \qquad t^{(2n- \frac 12)}\|B^{(n-1)}(t)\|_{H_1^A}^2 + & \int_0^t
             s^{2n- \frac 12} \|A^{(n)}(s) \|_{H_1^A}^2 \,ds              \le C_{n5}(t)     \label{BnH}
             \\
\tag{$\F_n$}  \qquad\ \ \  \  \ t^{2n+\frac 12} \| A^{(n)} (t)\|_{H_1^A}^2  + & \int_0^t
            s^{2n+\frac 12} \| B^{(n)}(s)\|_{H_1^A}^2 \,ds                   \le C_{n6}(t). \label{AnH}
\end{align*}
\end{corollary}
Theorem \ref{MainTh} and Corollary \ref{Maincor} will be proven in Section \ref{secpmt}.

\begin{remark}\label{remqs}{\rm (Analysis over quotient space)
Denote by $\Y$ the set of almost strong solutions of the
Yang-Mills heat equation over $M$
with initial value in $H_{1/2}$ and having finite action. The group $\G_{3/2}$ acts on $\Y$
through its action on $A(0)$ for each $A \in \Y$. For simplicity of statement let us assume that
uniqueness of solutions holds in this class. All of the functionals appearing
on both sides of the inequalities in Theorem \ref{MainTh} and Corollary \ref{Maincor}
descend to functions  of the initial values on  the quotient space ${\cal{C}} \equiv \Y/\G_{3/2}$.
 The theorem and its corollary can and should be interpreted
 as regularity properties of functions on the quotient  space.
 It will be shown in \cite{G72} that $\cal C$ is a complete metric space in a natural
metric.
}
\end{remark}

\section{The lower order terms}

Our strategy  consists in computing the gauge covariant exterior derivatives
and coderivatives of all the $n$th order time derivatives $A^{(n)}(t)$ and $B^{(n)}(t)$ and expressing
them in terms of lower order time derivatives. This will be done in the next subsection.
These identities, in turn, will give rise to integral identities, which will be used in Section \ref{secib}
 to establish initial behavior bounds by induction on $n$.

\subsection{Pointwise identities.}

In this section we assume that $A(t)$ is  a time dependent  $\kf$ valued connection form over $M$,
which is in $C^\infty((0,T)\times M)$ and solves the Yang-Mills heat equation \eref{I2}.
$B(t)$ denotes the curvature of $A(t)$. We will derive some identities by applying $d_A$ and $d_A^*$ to various $\kf$ valued forms.
In case $M \ne \R^3$  one needs to specify  boundary conditions on a $p$-form $\omega$ in order for it to belong to the
 domain of $d_A$ or $d_A^*$. These are  analogous to the Dirichlet and Neumann
  boundary conditions for the domain of $d$ and $d^*$ discussed in \cite{Co}. We recall from
   Section 3 of \cite{CG1} that for the Dirichlet boundary conditions, $(D)$, $d_A$ is the minimal
    operator. It imposes
nontrivial boundary conditions on the forms in its domain. $d_A^*$ is maximal in this case.
 On the other hand, for the Neumann boundary conditions,  $(N)$, $d_A$ is maximal and the domain of $d_A^*$
  imposes  nontrivial boundary conditions on its elements.

 The next proposition  expresses spatial derivatives of solutions
in terms of time derivatives.
\begin{proposition} \label{propA1}  Let $A(t)$ be a smooth solution to the Yang-Mills heat equation
over $(0,T)$,  satisfying either \eref{N2} or \eref{D1}
if $M\ne \R^3$.
Then there exist non-negative constants $c_{ni}, \bar{c}_{ni}, \tilde{c}_{ni}, \hat c_{ni}$,
that depend only on $n$ and $i$, such that,
 for all $n\geq 1$ and  $0 < t < T$, the following identities hold.
\begin{align}
d_{A(t)} A^{(n)}(t) &= B^{(n)}(t)  -P_n(t), \ \ \    \text{where}  \label{A11}    \\
           & P_n(t)=  \sum_{i=1}^{n-1} c_{ni} [A^{(i)}(t)\wedge A^{(n-i)}(t) ].  \notag \\
d_{A(t)}^*B^{(n-1)} \ &= -A^{(n)} (t)  - Q_n(t),\ \ \ \text{where}   \label{A12}   \\
        & Q_n(t) = \sum_{i=1}^{n-1} \bar{c}_{ni} [A^{(i)}(t)\lrc B^{(n-1-i)}(t)].  \notag\\
         d_{A(t)}^* A^{(n)}(t)\  &= - R_n(t),\ \ \   \text{where}  \label{A13}  \\
     & R_n(t) = \sum_{i=1}^{n-2} \tilde{c}_{ni} [A^{(i)}(t)\lrc A^{(n-i)}(t)] .  \notag \\
\text{Moreover, for} &\  \text{all $n\ge0$ there holds} \notag \\
     d_{A(t)} B^{(n)}(t) &= S_n(t),  \ \     \text{where}\ S_0(t) =0,\   S_1(t) = [B(t)\wedge A'(t)]\
                                                                \text{and}  \label{A14} \\
  &S_n(t) = [B(t)\wedge A^{(n)}(t)]
                          +\sum_{i=1}^{n-1} \hat c_{ni}  \, [(B^{(i)}(t) -P_i(t))\wedge A^{(n-i)} (t)]    \notag\\
        \text{for} \ \ n \ge 2.\ \  &  \notag
 \end{align}
The functions $P_n(t), Q_n(t), R_n(t)$ are polynomials in the time derivatives of $A$ and $B$
of order at most $n-1$ in $A$ and at most $n-2$ in $B$. Empty sums are to be interpreted as zero.
In particular,
\[
P_1(t) = Q_1(t) = R_1(t) = R_2(t)=0.
\]
In the above identities $d_A$ is the exterior derivative with domain matching the boundary conditions and $d_A^*$ is its adjoint.
\end{proposition}

The next lemma carries out the inductive computation, ignoring domain issues for the operators
$d_A$ and $d_A^*$. These issues, which are relevant if $M\ne \R^3$,
 will be addressed in the succeeding lemmas.

\begin{lemma}\label{lemA1} The identities  \eref{A11} - \eref{A14} hold, ignoring boundary conditions.
\end{lemma}
\begin{proof}
We will prove the identities \eref{A11}- \eref{A13} by induction on $n$.
   Recall the identity
   \beq
   d_A A' = B'   \label{A0}
   \eeq
   proved in \cite[Section 5]{CG1}, which is \eref{A11} for $n =1$, since $P_1(t) =0$.
Let $k \ge 1$. Assume that the identity \eref{A11} holds for $n=k$
and differentiate  both sides with
 respect to $t$   to find $(d_AA^{(k)})' = B^{(k+1)} - P_k'$. Therefore
 \begin{align*}
 d_A A^{(k+1)}& =   B^{(k+1)}  - [A'\wedge A^{(k)}] -\sum_{i=1}^{k-1} c_{ki} ([A^{(i)}\wedge A^{(k-i)}])' \\
&= B^{(k+1)}  - [A'\wedge A^{(k)}]
                 - \sum_{i=1}^{k-1} c_{ki} ([A^{(i)}\wedge A^{(k+1-i)}] +[A^{(i+1)}\wedge A^{(k-i)}] ).
 \end{align*}
Thus \eref{A11} holds  with $c_{(k+1)1}=1+c_{k1}$ and $c_{(k+1)i}=c_{k(i-1)}+c_{ki}$ for $2\leq i \leq k$.
 Notice that $[A^{(i)}\wedge A^{(j)}]=[A^{(j)}\wedge A^{(i)}]$ for any $i, j$.
 The coefficients $c_{ni}$ are
   the ones obtained from the inductive process above. This proves \eref{A11}.

To prove \eref{A12} observe that for $n=1$ this is the Yang-Mills heat equation
 since $Q_1(t) =0$.
For $n=2$, the identity  $d_A^* B' = -A''  - [A'\lrc B]$, proved in \cite[Section 5]{CG1}
gives \eref{A12} with  $\bar c_{21}=1$ .
Assume that \eref{A12} holds for $n=k\ge 2$ and differentiate both sides with respect to $t$ to obtain
$d_A^* B^{(k)} +[A'\lrc B^{(k-1)}] = -A^{(k+1)} - Q_k'$. Therefore
\begin{align*}
d_A^* B^{(k)} &= -A^{(k+1)}  - [A'\lrc B^{(k-1)}]  - \sum_{i=1}^{k-1} \bar{c}_{ki} ([A^{(i)}\lrc B^{(k-1-i)}])' \\
&=  -A^{(k+1)}  - [A'\lrc B^{(k-1)}]
- \sum_{i=1}^{k-1} \bar{c}_{ki} ( [A^{(i)}\lrc B^{(k-i)}] +[A^{(i+1)}\lrc B^{(k-1-i)}]).
\end{align*}
This is \eref{A12} with $n=k+1$ and coefficients given
by $ \bar{c}_{(k+1)1}=1+\bar{c}_{k1}$ and $\bar{c}_{(k+1)i}=\bar{c}_{ki} +\bar{c}_{k(i-1)}$
for $2\leq i \leq k$.

For the proof of \eref{A13} we observe that
\[
d_A^*A'= - d_A^* d_A^* B =0
\]
by the Bianchi identity.
Differentiating both sides with respect to $t$ we get
\[
0= (d_A^*A')'= d_A^*A'' + [A'\lrc A'] = d_A^*A''
\]
since $[\w \lrc \w]=0$ for any 1-form $\w$. Differentiating once again with respect to $t$ we obtain
\[
d_A^*A''' + [A'\lrc A''] =0.
\]
This proves \eref{A13} for $n=1$ and $n=2$ because $R_1 = R_2 =0$.
Let $k \ge 2$ and assume that
  \eref{A13} holds for $n=k$. Differentiate both sides with respect to $t$ to get
\[
d_A^* A^{(k+1)} + [A' \lrc A^{(k)}] +
\sum_{i=1}^{k-2} \tilde{c}_{ki} ([A^{(i)} \lrc A^{(k+1-i)} ] + \, [A^{(i+1)} \lrc A^{(k-i)} ] \, )=0.
\]
This is \eref{A13} with $n = k+1$.

Finally we will derive \eref{A14} by applying $d_A$ to both sides of \eref{A11} rather than
 proceeding by induction.  For $n=0$ the identity \eref{A14} is just the Bianchi identity.
 For $n \ge 1$  we find
\begin{equation}
d_A B^{(n)} = d_A d_A  A^{(n)}
                   +\sum_{i=1}^{n-1} c_{ni} d_A \bigl( \,[A^{(i)}\wedge A^{(n-i)} ]\, \bigr). \label{A30}
\end{equation}
By the Bianchi identity we have
 $d_A d_A  A^{(n)}=[B\wedge A^{(n)}]$. Moreover,
$
d_A[\omega\wedge \eta]=[d_A\omega\wedge \eta] -[\omega\wedge d_A \eta]
$
for 1-forms $\omega, \eta$ and
$
[u\wedge v] = - [v\wedge u]
$
 whenever $u$ is a $\kf$ valued 1-form and $v$ is a $\kf$ valued 2-form.
 Therefore \eref{A30} gives
\begin{align*}
d_A B^{(n)} & =[B\wedge A^{(n)}] + \sum_{i=1}^{n-1} c_{ni} \bigl\{ \, [d_A A^{(i)}\wedge A^{(n-i)} ] + [d_A A^{(n-i)} \wedge A^{(i)} ]\, \bigr\}\\
& = [B\wedge A^{(n)}] + \sum_{i=1}^{n-1} (c_{ni} + c_{n(n-i)} )\, [d_A A^{(i)}\wedge A^{(n-i)} ].
\end{align*}
Using \eref{A11} to substitute for the  term  $d_A A^{(i)}$ we arrive at \eref{A14} with
 $\hat c_{ni} = c_{ni} + c_{n(n-i)}$.
\end{proof}

\bigskip

Although we applied
the exterior derivative operator and its adjoint to smooth forms in the preceding lemma,
we need to verify that the boundary conditions satisfied by these forms
match with the domains of these operators when $M \ne \R^3$.
   To this end we recall here some properties
    of these domains, established in Section 3 of \cite{CG1}.
\begin{lemma}$($\textup{\cite[Lemma 3.4]{CG1}}$)$   \label{LemBdy1}
Suppose that $\w \in W_1(M; \L^p\otimes \frak k)$ and $A \in L^\infty(M)$.
Then
\begin{align*}
(D)\qquad  & \w \in Dom(d_A)\ \ \text{if and only if}\ \ \w_{tan}=0         \\
(N)\qquad  &  \w \in Dom(d_A^*)\ \ \text{if and only if}\ \ \w_{norm}=0.
\end{align*}
\end{lemma}

Moreover we proved the following:
\begin{lemma}$($\textup{\cite[Proposition 3.5]{CG1}}$)$ \label{LemBdy2}  Assume that $\w$
is a $\frak k$ valued form  and that $A \in W_1\cap L^\infty$.
 Denote the curvature  of $A$ by $B$, as in \eref{ymh3}.

\noindent
If $ [B\wedge \w] \in L^2$ then
\begin{align*}
(N)\ \ \ \w \in Dom(d_A)\ &\text{implies}\  \w \in Dom((d_A)^2)\
        \text{and}\  d_A^2 \w = [ B\wedge \w]
        \\
 \text{and}    \ (D)\  \ \     \w \in Dom(d_A)\ &\text{implies}\
            \w \in Dom((d_A)^2)\
       \text{and}\   d_A^2 \w = [ B\wedge \w].
        \end{align*}
  \noindent
 If $[B\lrc \w] \in L^2$ then
 \begin{align*}
 (D)\ \ \  \w \in Dom(d_A^*)\  &\text{implies}\     \w \in Dom((d_A^*)^2)\
      \text{and}\  (d_A^*)^2 \w = [ B\lrc \w]
      \\
\text{and}\ (N)\ \ \   \w \in Dom(d_A^*)\  & \text{implies}\    \w \in Dom((D_A^*)^2)
\ \text{and}\ \ \ (d_A^*)^2 \w = [ B\lrc \w].
\end{align*}
\end{lemma}

For the remainder
of this section we will assume that $A(t)\in C^{\infty}(\, (0,T)\times M: \Lambda^1\otimes \frak k)$ is a smooth solution to the Yang-Mills heat equation which satisfies  \eref{I2} and one of the boundary conditions \eref{N2} or \eref{D1} if $M \ne \R^3$.

\begin{lemma}\label{LemBdy3} Let $A(t)$ be a smooth solution to the Yang-Mills heat equation over $(0, T)$, satisfying either \eref{N2} or \eref{D1}.
 Denote by  $A^{(n)}(t)$, $B^{(n)}(t)$  the $n$th order
 time derivatives of $A$ and $B$  respectively.  If $A(\cdot)$ satisfies \eref{D1}
 then for all $n\geq 0$ and  $0 < t < T$
\begin{equation} \label{D8}
\begin{split}
A^{(n)}(t)_{tan} &= 0 \ \  \text{and}
                      \ \ A^{(n)}(t) \in Dom(d_A). \\
\ \ B^{(n)}(t)_{tan} &= 0 \ \  \text{and}
                    \ \ B^{(n)}(t) \in Dom(d_A).
\end{split}
\end{equation}
If $A(t)$ satisfies  \eref{N2},
 then for all $n\geq 1$ and  $0 < t < T$
\begin{equation}  \label{N8}
\begin{split}
B^{(n)}(t)_{norm} &= 0 \ \  \text{and}
 \ \ B^{(n)}(t) \in Dom(d_A^*). \\
\ \ A^{(n)}(t)_{norm} &= 0 \ \  \text{and}
\ \ A^{(n)}(t) \in Dom(d_A^*).
\end{split}
\end{equation}
\end{lemma}
\begin{proof}
We begin with the Dirichlet case. By \eref{D1}   we have
  $A(t)_{tan} = 0$ for all $t\in (0, T)$. We may differentiate $A(t)_{tan}$ with respect to $t$
  on the boundary to get $A^{(n)}(t)_{tan} = 0$ for all $n\geq 0$.
  Therefore,  $A^{(n)}(t)$  belongs to the domain of the minimal operator $d_A$ in this
   case by Lemma \ref{LemBdy1}.
   By Corollary 3.7 in \cite{CG1}, $A(t)_{tan} = 0$  also implies that $B(t)_{tan} = 0$.
   As a result,  $B^{(n)}(t)_{tan} = 0$ for all $n\geq 0$ and $B^{(n)}(t)$ therefore also belongs
 to the domain of $d_A$.

Similarly, in the Neumann case,  since  $B(t)_{norm}=0$ for all $t\in (0, T)$, it follows
 that $B^{(n)}(t)_{norm} = 0$ for all $t \in (0, T)$ and therefore $B^{(n)}(t)$  belongs to the
 domain of the minimal operator $d_A^*$ by Lemma \ref{LemBdy1}. By Lemma \ref{LemBdy2},  $B(t)_{norm}=0$ implies $d_A^* B(t)$ also
 belongs to the domain of $d_A^*$. Since $d_A^* B(t)=A'(t)$, we can apply Lemma \ref{LemBdy2}
  to find $A'(t)_{norm}=0$    for all $t\in (0, T)$ as well.
  As a result, $A^{(n)}(t)_{norm}=0$ for all $n\geq 1$ and  therefore $A^{(n)}(t)$
  also belongs to the domain of $d_A^*$.
\end{proof}

\begin{lemma}\label{lemBdy4} In case $M\ne \R^3$ the operators $d_A$ and $d_A^*$ act
 only on elements in their domains in the identities \eref{A11} - \eref{A14}.
\end{lemma}
        \begin{proof}
The proof is similar to the proof of Lemma 5.1 in \cite{CG1}. For the Dirichlet case,
\eref{D8}     implies that  for all $n\geq 1$ and $t\in (0, T)$, $A^{(n)}(t)$  belongs to the domain of the
 minimal operator $d_A$. This justifies the use of  $d_A$ in \eref{A11}.
  Similarly \eref{D8} shows
that $B^{(n)}(t)$ is in the domain of $d_A$, which justifies  its use in \eref{A14}.
Since  $d_A^*$ is the maximal operator, $B^{(n)}$ and $ A^{(n)}$ both
 belong to its domain. This justifies its use in  \eref{A12} and \eref{A13}.

For the Neumann case     \eref{N8} of
 Lemma \ref{LemBdy3} shows that  $B^{(n)}(t)$ and $A^{(n)}(t)$ belong to the domain
  of the minimal operator $d_A^*$ for all $n\geq 1$ and $t\in (0, T)$.
  Therefore the application of $d_A^*$ in \eref{A12} and \eref{A13} is justified.
  The application of $d_A$ in \eref{A11}   and \eref{A14} is also justified,
  since it is the maximal operator in this case.
  \end{proof}

\bigskip
\noindent
 \begin{proof}[Proof of Proposition \ref{propA1}]
For the case $M\ne \R^3$ the identities \eref{A11}-\eref{A13} are justified by proof of Lemma \ref{lemA1} and Lemma \ref{lemBdy4}. For \eref{A14} it suffices to justify the application of $d_A$ to both sides of \eref{A11} under both sets of boundary conditions. In the case of Dirichlet boundary conditions observe that, for  all $n \ge1$, $B^{(n)}(t) \in Dom(d_A(t))$, by Lemma \ref{LemBdy3},
  as is $d_{A(t)} A^{(n)}(t)$  by Lemmas  \ref{LemBdy3}  and  \ref{LemBdy2}.
  Moreover, since all  $A^{(i)}(t)_{tan} = A^{(n-i)}(t)_{tan}= 0$ and
  $[A^{(i)}(t)\wedge A^{(n-i)} (t)]_{tan}=0$, the application of $d_A$ to each term
   on the right side of \eref{A11} is justified.
   The Neumann case is  trivial because $d_A$ is the maximal operator.

For  $M= \R^3$ the identities are justified since we are considering smooth
 solutions to the Yang-Mills heat equation. Boundary conditions are not an issue.
\end{proof}

\subsection{Integral identities.}

We will use the pointwise identities of the previous subsection to prove integral
 identities for smooth solutions to the Yang-Mills heat equation.

\begin{lemma} \label{lemA2} Let $A(t)$ be a smooth solution to the Yang-Mills heat equation
over $(0,T)$,  satisfying either \eref{N2} or \eref{D1}
  if $M\ne \R^3$. Then, for any integer $n\geq 0$,
\begin{equation} \label{A21}
  \begin{split}
\frac{d}{dt} \|B^{(n)}(t)\|_2^2  & + \| A^{(n+1)}(t)\|_2^2     \\
&=  - \| d_A^*B^{(n)}(t)\|_2^2 + \| Q_{(n+1)}(t)\|_2^2 + 2 (P_{n+1}(t), B^{(n)}(t) )
  \end{split}
\end{equation}
and, for any integer $ n \ge 1$,
\begin{equation}\label{A22}
  \begin{split}
   \frac{d}{dt}  \|A^{(n)}(t)\|_2^2 & +  \|B^{(n)} (t) \|_2^2     \\
&=  - \| d_A A^{(n)} (t)\|_2^2 +\|P_n(t)\|_2^2 -2 (Q_{n+1}(t), A^{(n)}(t)).
  \end{split}
\end{equation}
$d_A$ represents  the exterior derivative with domain matching the boundary conditions and $d_A^*$ is its adjoint.
\end{lemma}
        \begin{proof}
By identity \eref{A11}
\begin{align*}
(d/dt) \|B^{(n)}\|_2^2 &= 2 (B^{(n+1)}, B^{(n)}) \\
&=  2 ( d_A A^{(n+1)} +P_{n+1}, \,  B^{(n)}) \\
& = 2 (A^{(n+1)} ,  d_A^*  B^{(n)}) + 2 (P_{n+1}, \,  B^{(n)}),
\end{align*}
where we observe that the integration by parts is justified for both
boundary conditions. The first term on the right side may be written in two different ways using \eref{A12}
\begin{align*}
(A^{(n+1)} ,  d_A^*  B^{(n)}) & = - ( d_A^*B^{(n)} , d_A^*  B^{(n)}) - (Q_{n+1}, d_A^*  B^{(n)}) \\
\text{and also}\ &= -(A^{(n+1)} , A^{(n+1)})  - (A^{(n+1)}, Q_{n+1}).
\end{align*}
Adding the two we obtain
\begin{align*}
2 (A^{(n+1)} ,  d_A^*  B^{(n)}) & = - \| A^{(n+1)}\|_2^2 - \| d_A^*B^{(n)}\|_2^2
           - (Q_{n+1}, A^{(n+1)} + d_A^*  B^{(n)}) \\
&= - \| A^{(n+1)}\|_2^2 - \| d_A^*B^{(n)}\|_2^2 +  \| Q_{n+1} \|_2^2,
\end{align*}
where for the last equality we have applied \eref{A12} once more. \eref{A21} follows.

The second identity is proved in a similar manner. Using \eref{A12}
\begin{align*}
(d/dt) \|A^{(n)}\|_2^2 &= 2 (A^{(n+1)}, A^{(n)}) \\
&=  - 2 (d_A^* B^{(n)} , \,  A^{(n)} ) - 2( Q_{n+1}, \,  A^{(n)}) \\
& = -  2 ( B^{(n)}  ,  d_A  A^{(n)}) - 2(Q_{n+1} , \,  A^{(n)}),
\end{align*}
noting that the integration by parts is again justified for both
 boundary conditions. We rewrite the first term in two different ways using \eref{A11}
\begin{align*}
( B^{(n)}  ,  d_A  A^{(n)}) & = ( d_A A^{(n)} , d_A  A^{(n)}) + (P_n  \, , d_A  A^{(n)}) \\
&= ( B^{(n)}  , B^{(n)})- (P_n\ , B^{(n)} ).
\end{align*}
Adding the two we obtain
\begin{align*}
2( B^{(n)}  ,  d_A  A^{(n)})
& =  \| B^{(n)} \|_2^2 +  \| d_A A^{(n)} \|_2^2   + (P_n \ ,d_A  A^{(n)} - B^{(n)}) \\
&= \|B^{(n)}  \|_2^2 + \| d_A A^{(n)} \|_2^2 - \|P_n \|_2^2
\end{align*}
by applying once again \eref{A11} for the last equality. \eref{A22} follows.
\end{proof}

\section{Differential inequalities} \label{secdiffineq}

\subsection{Gaffney-Friedrichs-Sobolev  inequalities in three dimensions.} \label{secGFS}

In our estimates, the embedding of  $W_1$  into $L^6$  will be critical.
Define the gauge invariant version of the $W_1$ norm on $M$ by
\[
\|\w\|_{W_1^A(M)}^2 = \|\n^A \w \|_{L^2(M)}^2 + \| \w \|_{L^2(M)}^2
\]
for any $\frak k$ valued $p$-form $\w$ on $M$.

On a compact three-dimensional manifold $M$  with smooth boundary, as well as on $\R^3$,
the Sobolev embedding theorem implies that for any $\w \in W_1(M)$
\[
\|\w \|_6^2 \le (\kappa^2/2) ( \int_M | grad |\w |\, |^2 + \|\w \|^2_2 )
\]
for some constant $\kappa$ that depends on the  geometry of $M$,  but not on $A$ (see for example, \cite[Theorem 7.26]{GT}.) It holds also for $M= \R^3$. In view of Kato's inequality,
\[
 \int_M | grad | \w|\, |^2 \le \| \n^A \w \|_2^2,
\]
it follows that
\begin{align}
\|\w \|_6^2 \le (\kappa^2/2) \|\w\|_{W_1^A(M)}^2 \
                            \   \text{for}\  \w\  \text{and}\ A \in W_1(M).    \label{gaf54}
\end{align}

We recall the following gauge invariant Gaffney-Friedrichs inequality
\begin{theorem}$($\textup{\cite[Theorem 2.17]{CG1}}$)$  \label{thmGF} Suppose  that $M$ is a compact three-dimensional Riemannian manifold with smooth boundary  or that $M = \R^3$
and that $A$ is a $\frak k$ valued 1-form in $W_1(M)$
  with curvature $B$ such that $\|B\|_2 < \infty$.
  Then there exist constants $\lambda_M$ and $\gamma$  that  depend only
   on the geometry of $M$ and not on $A$,
such that, for
\beq
\lambda(B) :=  \lambda_M + \gamma \| B\|_2^{4},   \label{gaf51}
\eeq
there holds
        \begin{align}
(1/2) \|\w\|_{W_1^A(M)} ^2
   \le  \|d_A \w \|_2^2  + \|d_A^* \w \|_2^2
   + \lambda(B) \| \w\|_2^2                                                   \label{gaf50}
   \end{align}
 for any $\frak k$ valued $p$-form $\w$ in $W_1(M)$ satisfying  either
\[
 \w_{tan} =0\ \ \ \text{or}\ \ \ \w_{norm} =0
\]
if $M \ne \R^3$. Here $d_A$ is the covariant exterior derivative with domain matching the boundary condition on $\w$
and $d_A^*$ is its adjoint.
  \end{theorem}

We recall from \cite{CG1} that $\gamma = (1/4) (3\kappa^2)^3 c^4$ where $c\equiv  \sup \{ \| ad\ x\|_{\frak k \rightarrow \frak k} : |x|_{\frak k} \le 1\}$ is a constant that measures the non-commutativity of $K$ and which is zero if $K$ is commutative. The constant $\kappa$ is the Sobolev constant from \eref{gaf54}.
       The constant $\lambda_M$ is given by
\[
 \lambda_M = 1 + \|W\|_\infty +\theta,
\]
where $W$ is the Weitzenb\"ock tensor on $p$-forms and $\theta$ is a constant determined
 by the lower bound of the second fundamental form on $\p M$. If $M$ is convex then we can take
 $\theta =0$ and if $M=\R^3$  or is a convex subset of $\R^3$ then we can take $\lambda_M = 1$.
Thus in this paper we take $\lambda_M = 1$.

\begin{corollary}\label{corGFS} $($Gaffney-Friedrichs-Sobolev inequality$)$
Suppose that $M= \R^3$ or $M$ is the closure of a bounded convex open subset of $\R^3$ with smooth boundary.
Let $A \in W_1(M)$  and suppose that $\|B\|_2 < \infty$.
If $\w$ is a $p$-form in $W_1(M) \cap  Dom(d_A)\cap Dom(d_A^*)$ then
\beq
\|\w \|_6^2
\le \kappa^2 ( \|d_A \w \|_2^2 + \|d_A^* \w \|_2^2
     + \lambda\| \w \| _2^2)                                  \label{gaf68}
\eeq
with $\lambda = \lambda(B) = 1 + \gamma \|B\|_2^4$.

Note: If $M \ne \R^3$ and $\w \in W_1(M)$ then the domain restrictions are equivalent
 to   $\w_{tan}= 0$ or $\w_{norm}=0$.
 \end{corollary}
       \begin{proof}
    Combine \eref{gaf50} and \eref{gaf54}.
\end{proof}

In the following lemma  lower order time derivatives are singled out in what is otherwise
the standard Gaffney-Friedrichs-Sobolev inequality. We will use the notation $H_1^A$ instead
of $W_1^A$ because the argument of these norms always satisfies the relevant boundary
 conditions when $M \ne \R^3$. Moreover agreement of time between the argument and $A$
 will also  be understood. Thus $\|A^{(n)}(t)\|_{H_1^A}^2 = \|\n^{A(t)}  A^{(n)}(t)\|_2^2 + \| A^{(n)}(t)\|_2^2$
 as in Notation \ref{ginvsob}.
 These Sobolev norms are gauge invariant.

\begin{lemma} \label{lemGFS} {\rm (GFS)}   Let $A(t)$ be a smooth solution to the
 Yang-Mills heat equation as in Proposition \ref{propA1}.
 Taking $\gamma$ as the constant defined after Theorem \ref{thmGF}, define
\beq
\lambda(t) = 1+ \gamma \|B(t)\|_2^4.     \label{gaf51t}
\eeq
Then  for any $n\geq 1$ we have
\begin{align}
\kappa^{-2}\|A^{(n)}(t)\|_6^2 &\le (1/2)\|A^{(n)}(t)\|_{H_1^A}^2                  \notag\\
&\le \| R_n(t) \|_2^2 + \| d_A A^{(n)}(t) \|_2^2
                        +\lambda(t) \|A^{(n)}(t)\|_2^2                                         \label{di35a} \\
             &\le   \| R_n(t) \|_2^2 + 2\|P_n(t)\|_2^2 + 2 \|B^{(n)}(t) \|_2^2
                             +  \lambda(t) \|A^{(n)}(t)\|_2^2.                              \label{di35b}
\end{align}
For any $\  n \ge 0$  we have
\begin{align}
\kappa^{-2}  & \| B^{(n)}(t)\|_6^2     \le (1/2)  \| B^{(n)}(t)\|_{H_1^A}^2          \notag\\
&\le  \| S_n(t)\|_2^2 +  \| d_A^* B^{(n)}(t)\|_2^2
            + \lambda(t) \|B^{(n)}(t)\|_2^2                                                      \label{di35c}\\
      &     \le  \| S_n(t)\|_2^2 + 2\|Q_{n+1}(t) \|_2^2 + 2 \|A^{(n+1)}(t)\|_2^2
            + \lambda(t) \|B^{(n)}(t)\|_2^2,                                         \label{di35d}
 \end{align}
  \end{lemma}
  \begin{proof}
Lemma \ref{LemBdy3} shows that for either boundary value problem, $A^{(n)}(t)$  satisfies
 the correct boundary condition that allows us to apply the   Gaffney-Friedrichs inequality  \eref{gaf50}.
 Using also the Sobolev inequality \eref{gaf54} we find
\begin{align}
\kappa^{-2}\|A^{(n)}(t)\|_6^2 &\le (1/2) \|A^{(n)}(t)\|_{H_1^A}^2   \notag \\
&\le   \| d_A A^{(n)}(t) \|_2^2 +  \| d_A^* A^{(n)}(t) \|_2^2
                        +\lambda(t) \|A^{(n)}(t)\|_2^2.                         \label{di37}
\end{align}
\eref{di35a} now follows from the identity \eref{A13}.    The identity \eref{A11} shows that
$\|d_A A^{(n)} \|_2^2 = \|B^{(n)} - P_n\|_2^2 \le 2 \|P_n\|_2^2  + 2\|B^{(n)}\|_2^2$, which
proves \eref{di35b}.

Similarly, the Sobolev inequality \eref{gaf54} and Gaffney-Friedrichs inequality \eref{gaf50}  show that
\begin{align}
\kappa^{-2}  \| B^{(n)}\|_6^2&\le (1/2) \| B^{(n)}(t)\|_{H_1^A}^2  \notag\\
&\le \|d_A  B^{(n)}\|_2^2 + \| d_A^* B^{(n)}\|_2^2
+ \lambda(t) \|B^{(n)}\|_2^2,                                  \label{di38}
\end{align}
which yields \eref{di35c} in view of the identity \eref{A14}. Moreover, in accordance with \eref{A12}
we have
$ \|d_A^* B^{(n)}\|_2^2 = \| A^{(n+1)} + Q_{n+1}\|_2^2 \le 2 \| Q_{n+1}\|_2^2 + 2 \| A^{(n+1)}\|_2^2$,
from which \eref{di35d} follows.
\end{proof}

\begin{remark}{\rm The Gaffney-Friedrichs-Sobolev inequalities take a very simple form
 in case $n=0$ or $1$. Thus we have
\begin{align}
 \kappa^{-2} \| B(t)\|_6^2 &\le \|A'(t)\|_2^2 + \lambda(t) \|B(t)\|_2^2, \label{di35b1} \\
\kappa^{-2}\|A'(t)\|_6^2 &\le  \| B'(t) \|_2^2
                        +\lambda(t) \|A'(t)\|_2^2\ \ \ \ \ \ \ \ \text{and}             \label{di35a1} \\
\kappa^{-2}  \| B'(t)\|_6^2    &\le \| \, [B(t) \wedge A'(t)]\, \|_2^2
                          +\| d_A^* B'(t) \|_2^2 + \lambda(t) \|B'(t)\|_2^2.            \label{di35c1}
 \end{align}
 The first of these follows from \eref{di38} with $n=0$ because $d_A B=0$ by the Bianchi identity
 and $d_A^* B = - A'$ by the Yang-Mills heat equation. The second follows directly form \eref{di37}
 because $d_A A' =B'$  and $d_A^* A'=0$. The third follows from \eref{di35c} because
 $S_1(t)= [B(t)\wedge A'(t)]$.
 }
 \end{remark}

\subsection{Differential inequalities.}
 For the remainder of this section we let $A(t)$ be a smooth solution to the
Yang-Mills heat equation as in Lemma \ref{lemA1}  and define $\lambda(t) $
as in   \eref{gaf51t}. We will be using the Gaffney-Friedrichs-Sobolev inequalities
  to estimate the right side of the integral identities of Lemma \ref{lemA2}.

\begin{lemma}\label{lemdi3}  $($Estimate of \eref{A22} for $n \ge 1$$)$.
For each integer $n \ge 1$
there is a constant $c_n$ depending only
on $n$, the manifold $M$ and the vector bundle $\mathcal V$ such that
\begin{equation} \label{di30}
\begin{split}
\frac{d}{dt} \| A^{(n)}(t)\|_2^2 + \|B^{(n)}(t)\|_2^2
&\le   \Big(\lambda(t) + c_n\|B(t)\|_2^4\Big)\| A^{(n)}(t)\|_2^2  + \| P_n(t)\|_2^2 \\
&\qquad  + 2 \kappa^2 \|\hat Q_{n+1}(t)\|_{6/5}^2 +  \| R_n(t) \|_2^2
\end{split}
\end{equation}
where $\hat Q_{n+1}(t)$ is defined in \eref{di33}.

Note: All time derivatives of $A$ or $B$ that occur in $P_n(t)$, $\hat Q_{n+1}(t)$ and $R_n(t)$ are of order less than $n$.
\end{lemma}
\begin{proof}  We need to bound the right hand side of the  integral equality \eref{A22}.
We will derive a bound for the last term  in \eref{A22} which will include a term that
 cancels with the term $- \| d_A A^{(n)}(t)\|_2^2$.  Define
\begin{align}
   \hat Q_{n+1}(t) &= \sum_{i=1}^{n-1} \overline c_{(n+1)i} [A^{(i)}(t) \lrc B^{(n-i)}(t)],
            \ \ \text{for}\ \ n\ge 2 \ \       \text{and}                                \label{di33}\\
   \hat Q_2(t) &= 0.  \notag
   \end{align}
Then \eref{A12}, with $n$ replaced by $n+1$, shows that
$Q_{n+1} = \hat Q_{n+1} +\bar c_{(n+1)n} [A^{(n)}\lrc B]$.
The only time derivatives $A^{(i)}$ in $\hat Q_{n+1}$ are  of order less than $n$.

 For non-negative functions $f$ and $g$ H\"older's inequality gives
$\|f^2 g\|_1= \|f^{3/2} f^{1/2} g\|_1 \le \|f^{3/2}\|_4 \|f^{1/2}\|_4 \|g\|_2 = \|f\|_6^{3/2} \|f\|_2^{1/2} \|g\|_2$.
Therefore,  for any $\epsilon >0$ we have
\begin{align}
\|f^2 g\|_1 &\le  \|f\|_6^{3/2} \|f\|_2^{1/2} \|g\|_2      \notag \\
&\le (3/4) (\epsilon^{-1}\|f\|_6^{3/2})^{4/3} + (1/4) (\epsilon \|f\|_2^{1/2} \|g\|_2)^4  \notag\\
&= (3/4) \epsilon^{-4/3} \|f\|_6^2   +(1/4) \epsilon^4 \| f\|_2^2 \| g\|_2^4.   \label{di39}
\end{align}

Let $c_n' = 2c \bar c_{(n+1)n}$. Then
\begin{align*}
   \Big| 2(Q_{n+1}, A^{(n)})\Big|
   &= 2 \Big| (\hat Q_{n+1}, A^{(n)}) + \bar c_{(n+1)n} ([A^{(n)} \lrc B], A^{(n)}) \Big| \\
   &\le \Big\{2 \| \hat Q_{n+1}\|_{6/5} \|A^{(n)}\|_6\Big\}  + \Big\{c_n' \| \, |A^{(n)}|^2 |B|\, \|_1\Big\} \\
   &\le \Big\{2 \kappa^2   \| \hat Q_{n+1}\|_{6/5}^2
       + (1/2)\kappa^{-2} \| A^{(n)}\|_6^2\Big\} \\
   &\qquad  +\Big\{ \frac{c_n'}{4} \,  \epsilon^4 \|B   \|_2^4\, \|A^{(n)} \|_2^2
   + (\frac{3c_n'}{4} \, \epsilon^{-4/3}\kappa^2) \, \kappa^{-2}\|A^{(n)}  \|_6^2\Big\},
\end{align*}
wherein we used \eref{di39} with $f = |A^{(n)}(t)|$ and $g = |B(t)|$.

Choose $\epsilon$ such that
    $(\frac{3c_n'}{4} \, \epsilon^{-4/3}\kappa^2) =1/2$. The two $\|A^{(n))}\|_6^2$ terms add to
    $\kappa^{-2}\|A^{(n))}\|_6^2$.
    Using the  Gaffney-Friedrichs-Sobolev inequality \eref{di35a}, we find
\begin{align*}
 &\Big| 2(Q_{n+1}, A^{(n)})\Big|
   \le  2 \kappa^2 \|\hat Q_{n+1}\|_{6/5}^2 + c_n\|B \|_2^4 \| A^{(n)}\|_2^2
  +\Big(\kappa^{-2} \| A^{(n)}\|_6^2\Big)                \notag  \\
 & \qquad \le  2 \kappa^2 \|\hat Q_{n+1}\|_{6/5}^2 + c_n\|B \|_2^4 \| A^{(n)}\|_2^2
 +  \Big( \| R_n  \|_2^2 + \| d_A A^{(n)}  \|_2^2  +\lambda(t) \|A^{(n)} \|_2^2 \Big),
 \end{align*}
 where $c_n =(c_n'/4) \epsilon^4 = (1/4) (3/2)^3 \kappa^6 (c_n')^4$.
  Insert this bound into \eref{A22}, canceling the terms  $\|d_A A^{(n)}\|_2^2$,   to find
     \begin{align*}
& \frac{d}{dt} \|A^{(n)}(t)\|_2^2 +  \|B^{(n)} (t) \|_2^2  \\
& \qquad \le  \|P_n(t)\|_2^2  + 2 \kappa^2 \|\hat Q_{n+1}\|_{6/5}^2 + c_n\|B \|_2^4 \| A^{(n)}\|_2^2
 +  \Big( \| R_n  \|_2^2  +\lambda(t) \|A^{(n)} \|_2^2 \Big)
\end{align*}
which is \eref{di30}.
\end{proof}

\begin{lemma}\label{lemdi4}
 $($Estimate of \eref{A21} for $n \ge 0$$)$ For each integer $n \ge0$ there holds
\begin{equation}  \label{di40}
  \begin{split}
&\frac{d}{dt}\|B^{(n)}(t)\|_2^2 + \| A^{(n+1)}(t)\|_2^2  \\
& \qquad \le \l(t) \| B^{(n)}(t)\|_2^2      +\|Q_{n+1}(t) \|_2^2 + \kappa^2 \|P_{n+1}(t)\|_{6/5}^2 + \| S_n(t)\|_2^2  .
  \end{split}
\end{equation}

Note: All time derivatives of $B$ that occur in $Q_n(t)$ are of order less than $n$.
All time derivatives of $A$ in the right side are of order less than $n+1$.
\end{lemma}
\begin{proof}  We  need to bound the right hand side of \eref{A21}. We have
\begin{align*}
      2|(P_{n+1}, B^{(n)})| &\le 2  \|P_{n+1}\|_{6/5} \| B^{(n)}\|_6  \\
      &\le \kappa^2   \|P_{n+1}\|_{6/5}^2 + \kappa^{-2}  \| B^{(n)}\|_6^2 \\
      &\le  \kappa^2   \|P_{n+1}\|_{6/5}^2 + \| S_n\|_2^2  + \| d_A^* B^{(n)}\|_2^2
 + \lambda(t) \|B^{(n)}\|_2^2
      \end{align*}
      by virtue of \eref{di35c}.
      Therefore
\[
 -\|d_A^*   B^{(n)}\|_2^2 +  2|(P_{n+1}, B^{(n)})| \le   \kappa^2   \|P_{n+1}\|_{6/5}^2 + \| S_n\|_2^2
 + \lambda(t) \|B^{(n)}\|_2^2
\]
 This proves \eref{di40}.
 \end{proof}

\begin{remark} {\rm In case $n=0$ the inequality \eref{di40} gives
\begin{align}
\frac{d}{dt} \| B(t)\|_2^2 + \| A'(t)\|_2^2 \le \lambda(t) \| B(t)\|_2^2
\end{align}
since $Q_1 = P_1 = S_0 = 0$ by Proposition \ref{propA1}. But the identity \eref{A21} shows that
$ \frac{d}{dt} \| B(t)\|_2^2 + 2\| A'(t)\|_2^2  = 0$.
There is a loss of information, therefore, in \eref{di40}, which we allow in order to get a simple
inequality for all $ n \ge 0$.
}
\end{remark}

\bigskip

 Under the assumption of finite action we will be able to use the preceding differential inequalities
  to obtain integral estimates in our main result, Theorem \ref{MainTh}.
 Proposition \ref{propdi6}
 below will be critical in this transition.

\begin{notation}\label{notdi5}{\rm  For a smooth solution $A(\cdot)$ on $(0,T)$
to the Yang-Mills heat equation \eref{I2}
  that has finite action let $c_n$ be the
 constant appearing in \eref{di30}  and define
 \begin{align}
 \psi(t)&=  \lambda_M\, t + \gamma\int_0^t \|B(\sigma)\|_2^4 d\sigma\ \ \ \text{and}  \label{di50}\\
 \psi_n(t) &=  \lambda_M\, t + (\gamma + c_n)\int_0^t \|B(\sigma)\|_2^4 d\sigma .   \label{di51}
\end{align}
Lemma 5.5 will show that $\int_0^t \|B(\sigma)\|_2^4 d\sigma < \infty$ when $A(\cdot)$ has finite action.
It follows from this that $\psi(t)$ and $\psi_n(t)$  are bounded, differentiable and
 nondecreasing functions on the interval $(0, T)$.
Then, for $0\le s \le t$ the functions
\[
\psi^{t,s} := \psi(t) - \psi(s)\ \  \text{and}\ \ \ \psi_n^{t,s} := \psi_n(t) - \psi_n(s)
\]
are non-negative.
 }
 \end{notation}

\begin{proposition}\label{propdi6}
\begin{align}
\frac{d}{ds}\Big(e^{-\psi_n(s)} \|A^{(n)}(s)\|_2^2\Big)
                     + e^{-\psi_n(s)}\|B^{(n)}(s)\|_2^2 \le  e^{-\psi_n(s)} X_n(s),\ \ \ n\ge 1  \label{di60}
\end{align}
and
\begin{align}
\frac{d}{ds}\Big(e^{-\psi(s)} \|B^{(n)}(s)\|_2^2 \Big)
         + e^{-\psi(s)}\|A^{(n+1)}(s)\|_2^2 \le e^{-\psi(s)} Y_n(s), \ \ \ n \ge 0, \label{di61}
\end{align}
where
\begin{align}
X_n(t) &=  \| P_n(t)\|_2^2 +2\kappa^2
       \|\hat Q_{n+1}(t)\|_{6/5}^2 +  \| R_n(t) \|_2^2 \ , \  n \ge 1 \ \ \  \text{and}        \label{di62}       \\
Y_n(t) &= \|Q_{n+1}(t) \|_2^2 + \kappa^2 \|P_{n+1}(t)\|_{6/5}^2 + \| S_n(t)\|_2^2, \ n \ge 0. \label{di63}
\end{align}
Note that $Y_0(t)= X_1(t) =0$ by virtue of Proposition  \ref{propA1} and the definition \eref{di33} of $\hat Q_2$.
\end{proposition}
\begin{proof}
Since $\psi'(s) = \lambda(s) $ and $\psi_n'(s) = \lambda(s) + c_n \|B(s)\|_2^4$,
the inequality \eref{di30} can be written as
\[
\frac{d}{ds}\|A^{(n)}(s)\|_2^2  -\psi_n'(s) \|A^{(n)}(s)\|_2^2 + \|B^{(n)}(s)\|_2^2  \le X_n(s).
\]
This is equivalent to \eref{di60}, as one can see by differentiating the product and then multiplying
by $e^{\psi_n(s)}$. The inequality  \eref{di61} follows from \eref{di40} similarly.
\end{proof}

\section{Initial behavior}      \label{secib}

\subsection{Initial behavior from differential inequalities.}

From the differential inequalities \eref{di60} and \eref{di61} we are going to derive
initial behavior bounds in the form of integral estimates with the help of the following
elementary lemma.
\begin{lemma}\label{lem50}$($Initial behavior from differential inequalities$)$
  Suppose that $f, g, h$ are nonnegative continuous
  functions on $(0, t]$ and that $f$ is differentiable. Suppose also that
 \beq
 (d/ds) f(s) + g(s) \le h(s)     ,\ \ \ 0<s \le t.                                         \label{vs90}
 \eeq
 Let $- 1 < b < \infty  $ and assume that
 \begin{align}
 \int_0^t s^{b} f(s) ds < \infty. \ \          \label{vs90f}
 \end{align}
 Then
 \beq
 t^{1+b} f(t) +\int_0^t s^{(1+b)}  g(s)ds \le \int_0^t s^{(1+b)}  h(s) ds
                   +(1+b) \int_0^t s^{b} f(s)ds.                                               \label{vs91}
 \eeq
 If equality holds in \eref{vs90} then equality holds in \eref{vs91}.
 \end{lemma}
 \begin{proof}
 Assumption \eref{vs90} implies that
 \[
 (d/ds) \left( s^{(1+b)} f(s) \right)  + s^{(1+b)} g(s) \le   s^{(1+b)} h(s)  + (1+b) s^{(1+b)} f(s)
 \]
 for all $ 0<s \le t.$ The result follows after integrating both sides over the interval $(0,t]$
 if one knows that $\lim_{t\downarrow 0} t^{(1+b)} f(t) =0$. See \cite[Lemma 4.8]{G70}
 for a proof without this assumption.
 \end{proof}

\begin{corollary} \label{corl50} Define $X_n(t)$ and $Y_n(t)$ by \eref{di62} and \eref{di63} respectively. The inequalities
\begin{align}
t^{2n -\frac{1}{2}} &\|A^{(n)}\|_2^2
        + \int_0^t
        s^{2n -\frac 12} \|B^{(n)}(s) \|_2^2 ds         \label{Ak2}\\
      &\le   \Big\{(2n-\frac 12)  \int_0^t
      s^{2n -\frac 32} \|A^{(n)}(s)\|_2^2 ds
      + \int_0^t
      s^{2n -\frac 12} X_n(s) ds\Big\}e^{\psi_n(t)}, \ \ \ n \ge 1 \notag
  \end{align}

  \begin{align}
 t^{2n+\frac 12} &\| B^{(n)} ( t)\|_2^2
         +  \int_0^t
         s^{2n+\frac 12} \| A^{(n+1)} (s) \|_2^2 \,ds    \label{Bk2}\\
 & \le  \Big\{(2n +\frac12) \int_0^t
 s^{2n -\frac12} \|B^{(n)}(s) \|_2^2 ds
         + \int_0^t
         s^{2n +\frac 12} Y_n(s) ds\Big\} e^{\psi(t)}, \ \ \ n \ge 0       \notag
  \end{align}
  hold whenever their right sides are finite,
  for $\psi(t), \psi_n(t)$ as in \eref{di50} and \eref{di51} respectively.
  \end{corollary}
 \begin{proof} Starting with the differential inequality  \eref{di60},
   we can apply Lemma \ref{lem50} with $ f(s) = e^{-\psi_n(s)} \|A^{(n)}(s)\|_2^2$,
    $g(s) = e^{-\psi_n(s)} \|B^{(n)}(s)\|_2^2$, $h(s) = e^{-\psi_n(s)} X_n(s)$ and $b = 2n -\frac 32$.
    Upon multiplying the resulting inequality \eref{vs91} by $e^{\psi_n(t)}$ we find
    \begin{align}
t^{2n -\frac{1}{2}} &\|A^{(n)}\|_2^2
        + \int_0^t e^{\psi_n^{t,s}}
        s^{2n -\frac 12} \|B^{(n)}(s) \|_2^2 ds         \label{Ak2b}\\
      &\le (2n-\frac 12)  \int_0^t e^{\psi_n^{t,s}}
      s^{2n -\frac 32} \|A^{(n)}(s)\|_2^2 ds
      +  \int_0^t e^{\psi_n^{t,s}}
      s^{2n -\frac 12} X_n(s) ds, \ \ \ n \ge 1. \notag
  \end{align}
  Since $1 \le e^{\psi_n^{t,s}} \le e^{\psi_n(t)}$,  the inequality \eref{Ak2b} continues to hold if we drop the
   factor $e^{\psi_n^{t,s}} $   from the integrand on the left and replace it in the
    integrands on the right by  $e^{\psi_n(t)} $.
   This yields \eref{Ak2}.

    The same method shows that \eref{Bk2} follows from \eref{di61} if, in Lemma \ref{lem50},
    one chooses  $f(s)=  e^{-\psi(s)} \|B^{(n)}(s)\|_2^2$, $g(s) =  e^{-\psi(s)} \|A^{(n+1)}(s)\|_2^2$,
    $h(s) = e^{-\psi(s)}Y_n(s)$ and $b = 2n-\frac 12.$
\end{proof}

\bigskip
The remainder of the paper will be devoted to proving that the right hand
 sides of the inequalities \eref{Ak2} and \eref{Bk2} are finite. This will be done by induction
 on $n$. But the induction hypothesis will include two other inequalities besides these two.

\subsection{Initial behavior of the curvature and $A'$.}

We review  a few well known apriori bounds for solutions of the Yang-Mills heat equation
in the presence of finite action.
\begin{lemma}  \label{lemFE}  Let $A(t)$ be a smooth solution to the Yang-Mills
heat equation  over  $(0, T)\times M$, satisfying either \eref{N2} or \eref{D1}
if $M\ne \R^3$.
Then $\| B(t)\|_2$ is nonincreasing on $(0,T)$.
Moreover, if $\|B_0\|_2 <\infty$ then
\beq
\| B(t)\|_2  \le \|B_0\|_2 \  \label{lemFEe1}
\eeq
for $0\le t <T.$
\end{lemma}
\begin{proof}
Identity \eref{A21} for $n=0$ gives $(d/dt) \| B(t)\|_2^2=-2 \| A'(t)\|_2^2 \leq 0$ since $P_1 = Q_1 =0$.
Therefore $\| B(t)\|_2^2$ is non-increasing. \eref{lemFEe1} follows from the continuity of $\|B(t)\|_2$
at $t =0 $ in this finite energy case.
\end{proof}

\begin{remark}{\rm  If $A(\cdot)$ has finite action  then $\rho(t)$, defined in \eref{dfa},  is finite
 for small $t$ and therefore   for all $t$, since the integrand is decreasing
   by Lemma \ref{lemFE}. Further, if $A(\cdot)$ is a solution
  with finite energy, i.e. $\|B_0\|_2 < \infty$, then   \eref{lemFEe1} shows that $A$ has finite action.
  }
  \end{remark}

\begin{proposition}\label{FA_Prop1}   Let $A(t)$ be a smooth solution to the Yang-Mills heat
 equation over $(0,T)\times M$, satisfying either \eref{N2} or \eref{D1}
 if $M \ne \R^3$.
  If $A(\cdot)$ has finite action then
\beq
t^{1/2} \| B(t) \|_2^2 + 2 \int_0^t s^{1/2}  \| A'(s) \|_2^2\, ds
= \rho(t),  \label{FA0}
\eeq
for any $t\in [0,T)$. In particular \eref{Bn1} holds for $n =0$.
\end{proposition}
\begin{proof}
For $n=0$ identity \eref{A21} becomes $(d/ds) \| B(s)\|_2^2 + 2 \| A'(s) \|_2^2 =0$.
We can apply Lemma \ref{lem50}, taking $f(s) = \| B(s)\|_2^2$, $g(s) = 2\|A'(s)\|_2^2$,
 $h(s) =0$ and $b =  -1/2$.
 Then \eref{vs91} asserts that
 \[
 t^{1/2} \| B(t) \|_2^2 + 2 \int_0^t s^{1/2}  \| A'(s) \|_2^2\, ds = (1/2)\int_0^t s^{-1/2} \|B(s)\|_2^2ds,
  \]
 which is \eref{FA0}, in view of the definition \eref{dfa} of $\rho(t)$.
  The hypothesis \eref{vs90f} is satisfied by the assumption of finite action. This proves that \eref{Bn1}
  holds for $n =0$. We can take $C_{03}(t) = \rho(t)$.
 \end{proof}

\begin{lemma}\label{lemFA1}   Let $A(t)$ be a smooth solution to the Yang-Mills
heat equation over  $(0,T)\times M$,
satisfying either \eref{N2} or \eref{D1}
 if $M\ne \R^3$.
If $A(\cdot)$ has finite action then
\begin{align}
t\|B(t)\|_2^4 &\le \rho(t)^2\qquad \text{and}            \label{FA1t}\\
\int_0^t   \| B(s) \|_2^4 \,ds &\le 2\, \rho(t)^2                    \label{FA1}.
\end{align}
Moreover, for $0 < t \le T$  there holds
\begin{align}
t\,\lambda(t)  &\leq \lambda_M t + \gamma \, \rho(t)^2,  \ \  \label{FA2t} \\
\psi(t)
                           &\le  \lambda_M t + 2 \gamma \, \rho(t)^2  \qquad  \text{and} \ \       \label{FA2}\\
\psi_n(t)  & \le  \lambda_M t + 2 (\gamma +c_n) \, \rho(t)^2    \label{FA2n}
\end{align}
where $\lambda(t)$ is defined in \eref{gaf51t} and $\psi(t)$ and $\psi_n(t)$ are defined in
\eref{di50} and \eref{di51} respectively.
In particular these three functions are non-decreasing and are bounded by standard dominating functions.
\end{lemma}
              \begin{proof}
Identity \eref{FA0} implies that $s^{1/2}   \| B(s) \|_2^2\leq \rho(s)\leq \rho(t)$
 for all $s\leq t$ since $\rho(t)$ is nondecreasing. In particular \eref{FA1t} holds. Further,
\begin{equation*}
\begin{split}
\int_0^t   \| B(s) \|_2^4 \,ds &= \int_0^t \left(s^{1/2}   \| B(s) \|_2^2 \right) \; \left(s^{-1/2}   \| B(s) \|_2^2\right)   \,ds \\
 &\leq \rho(t) \, \int_0^t s^{-1/2}   \| B(s) \|_2^2   \,ds = 2\, \rho(t)^2
\end{split}
\end{equation*}
proving \eref{FA1}.
The inequalities \eref{FA2t} - \eref{FA2n} now follow from their  definitions  and from \eref{FA1t} and  \eref{FA1}.
\end{proof}

\section{Proof of the Main Theorem} \label{secpmt}
\begin{remark} {\rm (Strategy.)
We will first prove the theorem under the technical assumption that $A(t)$ is a smooth
 solution to the Yang-Mills equation over $(0,T)\times M$. The proof will proceed by induction on $n$.
    We have already shown that \eref{Bn1} holds for $n =0$ in Proposition \ref{FA_Prop1}.
    We will
    show that all four inequalities
\eref{An1}, \eref{BnL6}, \eref{Bn1}, \eref{AnL6}
 in Theorem \ref{MainTh}  hold for $n =1$.
We will then show that if $k \ge 2$ and all four inequalities hold for
 $1 \le n  <k $ then all four inequalities hold for $n =k$. We will then remove the hypothesis of smoothness.
}
\end{remark}

\subsection{Proof for $n=1$.} \label{secn=1}

\begin{proposition}[Proof of $\A_1$]
\label{ord1prop}
Let $A(t)$ be a smooth solution to the Yang-Mills heat equation over $(0,T)\times M$, satisfying either \eref{N2} or \eref{D1}.
If $A(\cdot)$ has finite action for $0\leq t <T$ then
\beq
t^{3/2}\|A'(t)\|_2^2  + \int_0^t
s^{3/2} \|B'(s) \|_2^2 ds
              \le  C_{11}(t)   \label{eA1}
\eeq
for some standard dominating function $C_{11}$. In particular $\A_1$ holds.
\end{proposition}
              \begin{proof} Since $X_1(t) = 0$ the inequality \eref{Ak2} with $n = 1$ shows that
\[
t^{3/2}\|A'(t)\|_2^2  + \int_0^t
     s^{3/2}\|B'(s) \|_2^2 ds
               \le \frac 32 \,  e^{\psi_1(t)}\int_0^t  s^{1/2}\| A'(s)\|_2^2 \, ds\\
                \le  \frac 34 \,   e^{\psi_1(t)}\, \rho(t)
\]
wherein we have used  \eref{FA0} in the last inequality.
The bound \eref{FA2n} shows that the right hand side is bounded by a standard dominating function.
 \end{proof}

\bigskip
We see here that   the integrability of $t^{1/2}\| A'(t)\|_2^2$ in time implies the boundedness
 of $\,t^{3/2}\|A'(t)\|_2^2$ when $A(\cdot)$ is a solution to the Yang-Mills heat equation.
 This reflects a frequently occurring theme.

For the remainder of this section we will assume that $A(t)$ satisfies the assumptions of  Proposition \ref{ord1prop}.

\begin{corollary}[Proof of $\B_1$]  \label{ord1c1}
There exists a standard dominating function   $C_{12}$ such that
\begin{align}
t^{3/2}\|&B(t)\|_6^2 +\int_0^t  s^{3/2} \|A'(s) \|_6^2 ds \le C_{12} (t) \ \     \label{eB1}
\end{align}
for all $0\leq t <T$.
\end{corollary}
\begin{proof}
From the two GFS inequalities \eref{di35b1} and \eref{di35a1} we find
\begin{align*}
\kappa^{-2}&\Big( t^{3/2}\|B(t)\|_6^2 + \int_0^t
     s^{3/2} \|A'(s) \|_6^2 ds\Big) \\
&\le  \Big\{t^{3/2} \,\| A'(t) \|_2^2  +  (t \l(t)) t^{1/2} \|B(t)\|_2^2\Big\} \\
 & \qquad  + \int_0^t\Big\{  s^{3/2}  \|B'(s)\|_2^2 +(s\l(s))  s^{1/2} \|A'(s)\|_2^2\Big\} ds \\
 &\le \Big(t^{3/2} \,\| A'(t) \|_2^2   + \int_0^t  s^{3/2}  \|B'(s)\|_2^2\Big) \\
 & \qquad +t \l(t)\, \Big( \, t^{1/2} \|B(t)\|_2^2  + \int_0^t s^{1/2} \|A'(s)\|_2^2 ds\Big)\\
&\le C_{11}(t) +   \, t\lambda(t)\, \rho(t)
\end{align*}
wherein we have used \eref{eA1}, \eref{FA0}  and the nondecreasing property
of $t \lambda(t)$.  The bound \eref{FA2t} shows that $t\lambda(t) \rho(t)$ is bounded by a standard dominating function.
\end{proof}

\bigskip

To prove  $\C_1$   we need  the following integral estimate.

\begin{lemma}\label{ord1c3}
Define $Y_1(t)$ as in \eref{di63} with $n =1$. There is a  standard dominating function $\tilde C_{12}$ such that
 \beq
 \int_0^t s^{5/2} Y_1(s) ds \le   \tilde C_{12} (t)      \label{eB1a}
 \eeq
for all $0\leq t<T$.
\end{lemma}
\begin{proof} The definitions in Proposition \ref{propA1} give $Q_2(t) =\bar c_{21} [A'(t)\lrc B(t)]$ and
$P_2(t) = c_{21} [A'(t)\wedge A'(t)]$.  Hence, by the definition \eref{di63}  we have
 \begin{align}
 Y_1(t) &= \| Q_2(t) \|_2^2 + \kappa^2 \|P_2(t) \|_{6/5}^2 + \| S_1(t) \|_2^2 \notag\\
 &=\|\bar c_{21} [A'(t)\lrc B(t)]\, \|_2^2 + \kappa^2 \| c_{21} [A'(t)\wedge A'(t)]\, \|_{6/5}^2
                         + \|\, [ B(t) \wedge A'(t)]\, \|_2^2      \notag\\
   &\le \tilde c \| \ |A'(t)|\, |B(t)|\, \|_2^2 + \tilde c \|\, |A'(t) |^2\|_{6/5}^2    \label{121}
  \end{align}
 for some constant $\tilde c$ that only depends on the manifold and the bundle.
By H\"older's inequality
\[
\| \,| A'(s) |\, |B(s)|\, \|_2^2  \le   \|A'(s)\|_6^2 \|B(s)\|_3^2 \le   \|A'(s)\|_6^2 \|B(s) \|_2 \|B(s)\|_6 .
\]
Hence
\begin{align*}
s^{5/2}  \| \,| A'(s) |\, |B(s)|\, \|_2^2
   & \le     ( s^{3/2}  \|A'(s)\|_6^2 )\, ( s^{1/4} \|B(s)\|_2 )\, ( s^{3/4}  \|B(s) \|_6 ) \\
   & \leq   s^{3/2}  \|A'(s)\|_6^2  \,    \sqrt{\rho(t) \, C_{12}(t)}
\end{align*}
by \eref{FA0}     and  \eref{eB1}. Therefore
\begin{align}
\int_0^t
 s^{5/2} \| \,| A'(s) |\, |B(s)|\, \|_2^2\    ds & \leq  \sqrt{\rho(t) \, C_{12}(t)}   \int_0^t
  \,s^{3/2}  \|A'(s)\|_6^2 \;  ds              \notag\\
&\leq   \sqrt{\rho(t) \, (C_{12}(t)\,)^3} =:\tilde {\tilde C}_{12}(t)    \label{125}
\end{align}
by \eref{eB1},   giving an upper bound by a standard dominating function.

For the second term in \eref{121} we also apply H\"older's inequality twice to obtain
\[
\| \, |A'(s)|^2 \|_{6/5}^2   \le    \|  A'(s) \|_{3}^2 \|A'(s)\|_2^2  \leq \| \, A'(s) \|_6 \|  A'(s) \|_2 \|A'(s)\|_2^2 .
\]
Hence,
\begin{align*}
s^{5/2}  \| \, |A'(s)|^2 \|_{6/5}^2  & \le     ( s^{3/4} \| \, A'(s) \|_6 )\, ( s^{1/4} \|  A'(s) \|_2 )\, ( s^{3/2} \|A'(s)\|_2^2) \\
& \leq    ( s^{3/4} \| \, A'(s) \|_6 )\, ( s^{1/4} \|  A'(s) \|_2) \, C_{11}(t)
\end{align*}
by \eref{eA1}.
Therefore
\begin{align*}
\int_0^t &
s^{5/2} \| \, |A'(s)|^2 \|_{6/5}^2 ds \\
&\leq  C_{11}(t)\, \int_0^t
(s^{3/4} \| \, A'(s) \|_6 )\, ( s^{1/4} \|  A'(s) \|_2 ) \;  ds\\
&\leq  C_{11}(t)\, \left[ \int_0^t
s^{3/2} \| \, A'(s) \|_6^2
       \;  ds \right]^{1/2} \; \left[ \int_0^t
      s^{1/2} \|  A'(s) \|_2^2   \;  ds   \right]^{1/2}\\
&\leq    C_{11}(t)\,  \sqrt{ C_{12}(t) \, \rho(t) }
\end{align*}
by \eref{FA0}, \eref{eB1}   and   H\"older's inequality for the time integral.
\end{proof}

\bigskip
We are now ready to prove $\C_1$.

\begin{corollary}[Proof of $\C_1$]
 \label{ord2prop}
There is a standard dominating function $C_{13}$ such that
\begin{equation}
t^{5/2} \| B'(t) \|_2^2 + \int_0^t
 s^{5/2} \| A''(s) \|_2^2 ds
              \le  C_{13}(t)     \label{eC1}
\end{equation}
for all $0\leq t<T$.
\end{corollary}
\begin{proof}
From \eref{Bk2}  with $n=1$  we get
\begin{align*}
t^{5/2} \|B'(t)\|_2^2 + &\int_0^t s^{5/2}
   \| A''(s)\|_2^2 ds  \\
&\le  \Big\{\frac 52 \,  \int_0^t
           s^{3/2} \|B'(s) \|_2^2 \,ds  + \int_0^t
            s^{5/2} Y_1(s)\, ds\Big\} e^{\psi(t)}\\
            & \leq \Big\{\frac 52 \, C_{11} (t)  +
            \tilde C_{12}(t)\Big\} e^{\psi(t)}
\end{align*}
by \eref{eA1}   and \eref{eB1a}.
This is bounded by a standard dominating function in view of \eref{FA2}.
\end{proof}

\begin{corollary}[Proof of $\D_1$]
\label{ord2c1}
 There is a standard dominating function $C_{14}$ such that
\begin{equation}
  t^{5/2} \| A'(t)\|_6^2 + \int_0^t
   s^{5/2} \| B'(s)\|_6^2 ds
                         \le C_{14}(t)    \label{eD1}
\end{equation}
for all $0\leq t<T$.
\end{corollary}
\begin{proof} Multiply the GFS inequality \eref{di35a1} by $t^{5/2}$ to find
\begin{align*}
\kappa^{-2}t^{5/2} \, \| A'(t)\|_6^2 &\leq    t^{5/2} \,\|B'(t) \|_2^2 +   t  \l(t) \, ( t^{3/2} \| A'(t)\|_2^2)  \\
& \leq  C_{13}(t) + t\l(t) \,  C_{11}(t)
\end{align*}
by \eref{eC1}   and \eref{eA1}. This is bounded by a standard dominating function in view of \eref{FA2t}.

For the second term in \eref{eD1}
observe that the identity \eref{A12} reduces,
for $n =2$, to the identity $d_A^* B'=-A''-\bar c_{21} [A' \lrc B]$.
Replace $d_A^* B'$ by this in the GFS inequality \eref{di35c1} to find
\begin{align*}
\kappa^{-2} \| B'(t)\|_6^2 &\leq  \|\, [B\wedge A']\, \|_2^2 + 2 \|A''(t) \|_2^2
                            + 2 \bar c_{21}^2 \|\,  [A' \lrc B]\,  \|_2^2+\l(t) \| B'(t)\|_2^2 .
\end{align*}
 It follows  that for some constant $\bar c$ that only depends on the manifold and the bundle,
\begin{align*}
\int_0^t
 s^{5/2} \| B'(s)\|_6^2 ds
 & \le \bar c \int_0^t
  s^{5/2} \Bigl\{ \|A''(s) \|_2^2 +  \| \,|A'(s)|\,|B(s)|\, \|_2^2  +\l(s) \|B'(s)\|_2^2 \Bigr\}\, ds\\
  &\le   \bar c \, \Big\{\, C_{13}(t) +   \tilde {\tilde C}_{12}(t)
  +t\l(t) \int_0^t s^{3/2} \|B'(s)\|_2^2 ds\Big\} \\
  & \le   \bar c \, \Big\{\, C_{13}(t) +   \tilde {\tilde C}_{12}(t)  +t\l(t) C_{11}(t)\Big\}
\end{align*}
by \eref{eC1}, \eref{125},  and \eref{eA1}.
This is bounded by a standard dominating function in view of \eref{FA2t}.
\end{proof}

\bigskip
This completes the proof of Theorem \ref{MainTh}  for $n =1$ when $A(t)$ is smooth.

\subsection{Bounds on lower order terms.}

The induction mechanism in the next section will give us  information about the initial behavior of
the time derivatives of $A$ and $B$. We will use this information with the help of the following
bounds.

\begin{lemma} \label{lembds} (Bounds on lower order terms)   For all $n \ge 1$
there exist constants $d_{n,r}$ independent of $M$ and $A$ such that
\begin{align}
\| P_n(t)\|_2^2
&\le  d_{n,1}  c^2\sum_{i=1}^{n-1} \| A^{(i)}(t)\|_2  \| A^{(i)}(t)\|_6
                                         \|A^{(n-i)}(t) ]\|_6^2                                        \label{ib100}\\
\|P_{n}(t)\|_{6/5}^2
 &\le d_{n,2} c^2  \sum_{i=1}^{n-1} \| A^{(i)}(t)\|_6  \| A^{(i)}(t)\|_2\|A^{(n-i)}(t) \|_2^2  \label{ib101}\\
\|Q_n(t)\|_2^2
&\le d_{n,3}c^2  \sum_{i=1}^{n-1} \| A^{(i)}(t)\|_6^2 \| B^{(n-1-i)}(t)\|_2\| B^{(n-1-i)}(t)\|_6 \label{ib102} \\
\|\hat Q_{n}(t)\|_{6/5}^2
&\le d_{n,4}c^2  \sum_{i=1}^{n-2} \| A^{(i)}(t)\|_2 \| A^{(i)}(t)\|_6 \|  B^{(n-1-i)}(t)\|_2^2  \label{ib103} \\
\|R_n(t)\|_2^2
&\le  d_{n,5}  c^2 \sum_{i=1}^{n-2} \|A^{(i)}(t)\|_2 \| A^{(i)}(t)\|_6 \|A^{(n-i)}(t)  \|_6^2     \label{ib105} \\
\|S_n(t)\|_2^2
&\le   d_{n,6} \Big(\sum_{i =1}^n \|\, [A^{(i)}(t) \wedge B^{(n-i)}(t)]\, \|_2^2
               + \sum_{i=1}^{n-1} \|\, [ A^{(i)} (t)\wedge P_{n-i}(t)]\,\|_2^2 \Big). \label{ib107}
\end{align}
 Note: It will be clear from the proof that $d_{n,2}=d_{n,1}$ and $d_{n,4}\leq d_{n,3}$.
\end{lemma}
\begin{proof}
The Lemma is a simple application of H\"older's inequality.  From \eref{A11} we have
\begin{align*}
\|P_n(t)\|_2^2 &=  \| \sum_{i=1}^{n-1} c_{ni} [A^{(i)}(t)\wedge A^{(n-i)}(t) ]\, \|_2^2
 \le  d_{n,1} \sum_{i=1}^{n-1}  \|\,[A^{(i)}(t)\wedge A^{(n-i)}(t) ]\|_2^2 \\
&\le d_{n,1}  c^2\sum_{i=1}^{n-1} \| A^{(i)}(t)\|_3^2  \|A^{(n-i)}(t) ]\|_6^2  \\
&\le  d_{n,1}  c^2\sum_{i=1}^{n-1} \| A^{(i)}(t)\|_2  \| A^{(i)}(t)\|_6 \|A^{(n-i)}(t) ]\|_6^2 .
\end{align*}
This proves \eref{ib100}. For the second estimate we have
\begin{align*}
\|P_{n}(t)\|_{6/5}^2 &=  \|\sum_{i=1}^{n-1} c_{ni} [A^{(i)}(t)\wedge A^{(n-i)}(t) ] \|_{6/5}^2 \\
&\le   d_{n,2} \sum_{i=1}^{n-1} \| \,  [A^{(i)}(t)\wedge A^{(n-i)}(t) ]\,  \|_{6/5}^2 \\
&\le   d_{n,2} c^2  \sum_{i=1}^{n-1} \| A^{(i)}(t)\|_3^2 \|A^{(n-i)}(t) \|_2^2 \\
&\le   d_{n,2} c^2  \sum_{i=1}^{n-1} \| A^{(i)}(t)\|_6  \| A^{(i)}(t)\|_2\|A^{(n-i)}(t) \|_2^2.
\end{align*}

Similarly, from \eref{A12}
\begin{align*}
\|Q_{n}(t)\|_2^2 &= \| \sum_{i=1}^{n-1} \bar c_{ni}[A^{(i)}(t) \lrc B^{(n-1-i)}(t)]\, \|_2^2
 \le d_{n,3} \sum_{i=1}^{n-1} \|\, [A^{(i)}(t) \lrc B^{(n-1-i)}(t)]\, \|_2^2\\
&\le d_{n,3}c^2  \sum_{i=1}^{n-1} \| A^{(i)}(t)\|_6^2 \| B^{(n-1-i)}(t)\|_3^2 \\
&\le d_{n,3}c^2  \sum_{i=1}^{n-1} \| A^{(i)}(t)\|_6^2 \| B^{(n-1-i)}(t)\|_2\| B^{(n-1-i)}(t)\|_6.
\end{align*}
This proves \eref{ib102}.
From the definition \eref{di33} we find, for $n \ge 3$
\begin{align*}
\|\hat Q_{n}(t)\|_{6/5}^2
&= \|\sum_{i=1}^{n-2} \overline c_{ni} [A^{(i)}(t) \lrc B^{(n-1-i)}(t)]\, \|_{6/5}^2 \\
& \le d_{n,4} \sum_{i=1}^{n-2} \|\,  [A^{(i)}(t) \lrc B^{(n-1-i)}(t)]\, \|_{6/5}^2 \\
& \le d_{n,4}c^2  \sum_{i=1}^{n-2} \| A^{(i)}(t)\|_3^2 \|  B^{(n-1-i)}(t)\|_2^2  \\
& \le d_{n,4}c^2  \sum_{i=1}^{n-2} \| A^{(i)}(t)\|_2 \| A^{(i)}(t)\|_6 \|  B^{(n-1-i)}(t)\|_2^2
\end{align*}
proving \eref{ib103}.

From \eref{A13} we have
\begin{align*}
\| R_n(t)\|_2^2  &= \|\sum_{i=1}^{n-2} \tilde{c}_{ni} [A^{(i)}(t)\lrc A^{(n-i)}(t)]\,  \|_2^2
 \le d_{n,5}   \sum_{i=1}^{n-2} \|\, [A^{(i)}(t)\lrc A^{(n-i)}(t)]\,  \|_2^2  \\
&\le d_{n,5}  c^2 \sum_{i=1}^{n-2} \|A^{(i)}(t)\|_2 \| A^{(i)}\|_6 \|A^{(n-i)}(t)  \|_6^2,
\end{align*}
proving \eref{ib105}.

Finally, from \eref{A14} we have
\begin{align*}
\|S_n(t)\|_2^2  &=  \|\, [B(t)\wedge A^{(n)}(t)]
                  +\sum_{i=1}^{n-1} \hat c_{ni}  \, [(B^{(i)}(t) -P_i(t))\wedge A^{(n-i)} (t)]\, \|_2^2 \\
& \leq d_{n,6}'   \Big\{  \|\, [B(t)\wedge A^{(n)}(t)]\, \|_2^2
           +\sum_{i=1}^{n-1} \|\, [A^{(i)}(t) \wedge (B^{(n-i)}(t)-P_{n-i}(t) ) ] \,\|_2^2 \Big\} \\
           &\le d_{n,6} \Big(\sum_{i =1}^n \|\, [A^{(i)} \wedge B^{(n-i)}]\, \|_2^2
               + \sum_{i=1}^{n-1} \| \, [A^{(i)} \wedge P_{n-i}]\,\|_2^2 \Big) ,
 \end{align*}
proving \eref{ib107}.
\end{proof}

\subsection{Proof of the induction step.}

In Section \ref{secn=1} we proved  the  four inequalities of Theorem  \ref{MainTh}    for $n=1$.
 In this subsection we will assume that $k \ge 2$ and that the four  inequalities
 \eref{An1}, \eref{BnL6}, \eref{Bn1}, \eref{AnL6}
 of Theorem \ref{MainTh}  hold   for $1\le n <k$.
 We will   prove that they then also hold for $n=k$. For this purpose, we will need to show that
 the integrals involving $X_n$ and $Y_n$ in the inequalities \eref{Ak2} and \eref{Bk2} are finite
 under this induction hypothesis.  As in Section \ref{secn=1}, we will initially assume that $A(t)$ is smooth over $(0,T)\times M$.

\begin{lemma} \label{lemind1} If in Theorem \ref{MainTh}   the inequalities \eref{An1}, \eref{BnL6}, \eref{Bn1}, \eref{AnL6}
hold for  $1 \le n < k$ then
\begin{align}
\int_0^t
s^{2k -\frac12} X_k(s) ds & \leq \bar C_{k1}(t)  \ \ \ \ \text{and}      \label{ib130} \\
\sup_{0< t < T} t^{2k +\frac12} X_k(t) & \leq    \bar C_{k1}(t)           \label{ib131}
\end{align}
for some standard dominating function $\bar C_{k1}$.
\end{lemma}
\begin{proof}
 For the  proof of \eref{ib130}  it suffices to show that
\begin{align}
\int_0^t s^{2k -\frac12} \Big( \| P_k(s)\|_2^2 +
2\kappa^2 \|\hat Q_{k+1}(s)\|_{6/5}^2
                                                            +  \| R_k(s) \|_2^2\Big) ds  \leq \tilde  C_{k1}(t)        \label{ib133}
  \end{align}
by virtue of the definition \eref{di62} for $X_k$.
In view of  the inequalities
 \eref{ib100}, \eref{ib105} with $n =k$ and  \eref{ib103}  with $n = k+1$,  we need only show  that
\[
 \int_0^t s^{2k -\frac12} \|A^{(i)}(s)\|_2 \| A^{(i)}(s)\|_6
                     \Big(\|A^{(k-i)}(s)\|_6^2 + \| B^{(k-i)}(s)\|_2^2 \Big)   ds  \leq \tilde C_{k1}(t)
\]
for $1 \le i \le k-1$ and for some standard dominating functions $\tilde C_{k1}$.
 But for $1 \leq i \leq k-1$, the inductive hypotheses $\A_i$ and  $\D_i$
 of Theorem \ref{MainTh}  hold. Hence
\begin{equation} \label{na8}
\begin{split}
s^{2k-\frac 12} & \| A^{(i)} (s) \|_2 \; \| A^{(i)} (s) \|_6 \; \left\{\| A^{(k-i)}(s) \|_6^2   +   \| B^{(k-i)} (s) \|_2^2  \right\}   \\
& =  ( s^{i-\frac 14} \| A^{(i)} (s) \|_2 ) \;
(s^{i+\frac 14} \| A^{(i)} (s) \|_6)  \\
& \qquad \cdot \left\{ s^{2(k-i)-\frac 12} \,\| A^{(k-i)}(s) \|_6^2   +  s^{2(k-i)-\frac 12} \, \| B^{(k-i)} (s) \|_2^2  \right\}  \\
&\leq  \sqrt{  C_{i1}(t) } \,
\sqrt{ C_{i4}(t) }\,  \; \left\{ s^{2(k-i)-\frac 12} \,\| A^{(k-i)}(s) \|_6^2   +  s^{2(k-i)-\frac 12} \, \| B^{(k-i)} (s) \|_2^2  \right\}.
\end{split}
\end{equation}
The  factor in braces is integrable over $(0,t)$  by the inductive hypotheses $\B_{k-i}$ and $\A_{k-i}$ of
Theorem \ref{MainTh},  since  $k-i < k$. This proves \eref{ib130}.

For the proof of \eref{ib131} multiply the last line of \eref{na8} by $s$ and set $s=t$ to find
   \begin{align*}
t^{2k+\frac 12} & \| A^{(i)} (t) \|_2 \; \| A^{(i)} (t) \|_6
      \; \left\{\| A^{(k-i)}(t) \|_6^2   +   \| B^{(k-i)} (t) \|_2^2  \right\}     \\
& \le  \sqrt{  C_{i1}(t) } \, \sqrt{ C_{i4}(t) }\,
 \; \left\{ t^{2(k-i)+\frac 12} \,\| A^{(k-i)}(t) \|_6^2   +  t^{2(k-i)+\frac 12} \, \| B^{(k-i)} (t) \|_2^2  \right\} \\
 &\leq \sqrt{  C_{i1}(t) } \, \sqrt{ C_{i4}(t) }\,\left\{ C_{(k-i)4}(t) + C_{(k-i)3}(t)  \right\}
\end{align*}
where we have used the inductive hypothesis $\D_{k-i}$ of Theorem \ref{MainTh}  for the first summand in braces, and the inductive hypothesis
   $\C_{k-i}$ of Theorem \ref{MainTh},   for the second term.  These hold because $k - i < k$.
 Using the inequalities \eref{ib100}, \eref{ib103} and \eref{ib105} as before, we conclude that
 \begin{align}
t^{2k +\frac12}\Big( \| P_k(t)\|_2^2 +2\kappa^2
\|\hat Q_{k+1}(t)\|_{6/5}^2
                                                            +  \| R_k(t) \|_2^2\Big) \leq  \tilde C_{k1}(t).       \label{ib137}
 \end{align}
 This completes the proof of Lemma \ref{lemind1}.
 \end{proof}

\begin{proposition}[$\A_k$ holds]
\label{p.induct1}
 Let $A(t)$ be a smooth solution to the Yang-Mills heat equation over $(0,T)\times M$
with finite action and satisfying either \eref{N2} or \eref{D1} when $M \ne \R^3$.
 Assume that \eref{An1}, \eref{BnL6}, \eref{Bn1}, \eref{AnL6}
 hold for  $1 \le n < k$.
Then there exists a standard dominating function  $C_{k1}$ such that $\A_k$ holds.
\end{proposition}
\begin{proof}
Take $n=k$ in \eref{Ak2} to find
\begin{align*}
t^{2k -\frac{1}{2}} &\|A^{(k)}\|_2^2
        + \int_0^t
                        s^{2k -\frac 12} \|B^{(k)}(s) \|_2^2 ds        \notag\\
      &\le \Big\{(2k-\frac 12) \int_0^t
                                         s^{2k -\frac 32} \|A^{(k)}(s)\|_2^2 ds
      + \int_0^t
                     s^{2k -\frac 12} X_k(s) ds\Big\}e^{\psi_k(t)}    \\
      &\leq \Big\{(2k-\frac 12) \, C_{(k-1)3}(t) +  \bar C_{k1}(t)\Big\} e^{\psi_k(t)}
  \end{align*}
where we have used the inductive hypothesis   $\C_{k-1}$
to bound the first term on the right, and Lemma \ref{lemind1}  to bound the second term.
 Using \eref{FA2n} it follows that  there is a standard dominating function $C_{k1}$
 for which $\A_k$ holds.
\end{proof}

\begin{proposition}[$\B_k$ holds]
\label{p.induct2}  Let $A(t)$ as in Proposition \ref{p.induct1}.
   If in Theorem \ref{MainTh}, \eref{An1}, \eref{BnL6}, \eref{Bn1}, \eref{AnL6}
 hold for $n <k$
 then $\B_k$ holds for some standard dominating function $C_{k2}$.
\end{proposition}
\begin{proof}
From \eref{di35d} with $n$ replaced by $k-1$ we find
\begin{align}
\kappa^{-2}  \| B^{(k-1)}(t)\|_6^2 &     \le  \| S_{k-1}\|_2^2 + 2\|Q_{k} \|_2^2 + 2 \|A^{(k)}\|_2^2
            + \lambda(t) \|B^{(k-1)}\|_2^2 .         \label{di45d}
\end{align}
To prove boundedness of the first term in  $\B_k$
it suffices therefore to show that
\begin{align} \label{di35f}
t^{2k -\frac12} \Big(\| S_{k-1}\|_2^2 + 2\|Q_{k} \|_2^2 + 2 \|A^{(k)}\|_2^2
            + \lambda(t) \|B^{(k-1)}\|_2^2 \Big) \leq \tilde C_{k2}(t).
\end{align}
Concerning the second term in \eref{di35f},  observe that, for $1 \le i <k$, there holds
\begin{align*}
t^{2k -\frac12} & \| A^{(i)}(t)\|_6^2 \| B^{(k-1-i)}(t)\|_2\| B^{(k-1-i)}(t)\|_6 \\
&=\Big(t^{2i+ \frac12} \| A^{(i)}(t)\|_6^2\Big) \Big( t^{k-i -\frac34}  \| B^{(k-1-i)}(t)\|_2\Big)
 \Big(t^{k-i -\frac14}\| B^{(k-1-i)}(t)\|_6\Big) \\
&\leq C_{i4}(t)  \sqrt{C_{(k-1-i)3}(t) \, C_{(k-i)2}(t)},
\end{align*}
where we have used the inductive assumption  $D_i$ with $i < k$
in the first factor, the inductive assumption $\C_{k -1 -i}$ with $k-1-i<k$
in the second factor, and the inductive assumption $\B_{k-i}$ with $k-i<k$
in the the third factor. It follows from \eref{ib102}
 that $t^{2k -\frac12} \|Q_k(t)\|_2^2$ is bounded on $(0, T)$.

 By \eref{ib107} the first sum in $\|S_{k-1}(t)\|_2^2$ has similar
  bounds   as the terms in $\|Q_k(t)\|_2^2$ since
  $\|\, [ A^{(i)} \wedge B^{(k-1-i)}]\, \|_2^2\le c^2 \| A^{(i)}\|_6^2  \| B^{(k-1-i)}\|_2  \| B^{(k-1-i)}\|_6$,
  just as in the proof of \eref{ib102}.
             Therefore we need only address the terms of the form
$\|\, [A^{(i)}(t)\wedge P_{k-1-i}(t)]\, \|_2^2$   in \eref{ib107}  for $1 \le i \le k-2$. Replace $n$ by $k-1-i$ in the definition \eref{A11} to find
\begin{align*}
P_{k-1-i}(s) = \sum_{j=1}^{k-i-2} c_{(k-1-i)j}   [A^{(j)}(s)\wedge A^{(k-1-i-j)}(s) ]
\end{align*}
In view of \eref{ib107} it suffices to show that
\begin{align*}
t^{2k -\frac12}\|\, [  A^{(i)}(t) \wedge [ A^{(j)}(t)\wedge A^{(k-1-i-j)}(t) ]\,]\, \|_2^2
\end{align*}
is bounded on $(0, T)$ for  $1 \le i \le k-2$ and $ 1 \le j \le k-i-2$.        But
\begin{equation*}
\begin{split}
t^{2k-\frac 12} & \|\, |A^{(i)}(t)|\, |A^{(j)}(t)|\,|A^{(k-1-i-j)}(t)|   \,\|_2^2 \\
& \leq   \Big(t^{2i+\frac 12} \| A^{(i)}(t)\|_6^2\Big) \,\Big(t^{2j+\frac 12}\|A^{(j)}(t)\|_6^2\Big)
 \,\Big(t^{2k-2-2i-2j+\frac 12} \| A^{(k-1-i-j)}(t)\|_6^2\Big)\\
& \leq \, C_{i4} (t)  \,  C_{j4}(t) \, C_{(k-1-i-j)4}  (t)
\end{split}
\end{equation*}
by the induction hypothesis \eref{AnL6} for various values of $n <k$,
since $i,j\leq k-2$ and $k-1-i-j\leq k-3$. This gives us the boundedness of the first term in \eref{di35f}.

For the third term in \eref{di35f}, we use the inequality  $\A_k$ of Theorem \ref{MainTh},
which has already been proven in Proposition \ref{p.induct1}, to find
\[
t^{2k-\frac 12}  \| A^{(k)} (t) \|_2^2 \leq    \, C_{k1}(t).
\]

Finally,
\[
t^{2k - \frac12} \lambda(t)\|B^{(k-1)}(t)\|_2^2
=\Big(t \lambda(t)\Big) \Big( t^{2(k-1) +\frac12}  \|B^{(k-1)}(t)\|_2^2 \Big),
\]
which is a product of a bounded function, in accordance with \eref{FA2t} and another bounded function,
in accordance with the induction hypothesis $\C_{k-1}$.
Their product is bounded by a standard dominating function  by the usual argument.

We now turn our attention to the integral term of $\B_k$.
We need to prove that
\begin{align}
\int_0^t  s^{2k-\frac 12} \|A^{(k)}(s) \|_6^2 ds \leq \tilde C_{k2}(t)    \label{di45b}
\end{align}
for some standard dominating function $\tilde C_{k2}$.
By the inequality \eref{di35b} it suffices to find $\tilde C_{k2}$ such that
\begin{align*}
\int_0^t     s^{2k-\frac 12}\Big(  \lambda(s) \|A^{(k)}(s)\|_2^2 + 2 \|B^{(k)}(s) \|_2^2  +
                     \| R_k(s) \|_2^2 + 2\|P_k(s)\|_2^2    \Big) ds \leq  \tilde C_{k2}(t).
 \end{align*}
Now
\begin{align*}
\int_0^t  s^{2k-\frac 12}\lambda(s) \|A^{(k)}(s)\|_2^2 ds
&=\int_0^t\Big( s \lambda(s)\Big) s^{2k -\frac32} \|A^{(k)}(s)\|_2^2ds\\
&\leq t \lambda(t) \, C_{(k-1)3} (t)
\end{align*}
by  the inductive hypothesis $\C_{k-1}$.
Moreover $\int_0^t s^{2k-\frac 12}\|B^{(k)}(s) \|_2^2  ds \leq C_{k1}(t)$ by $\A_k$,
whose validity has been proven in Proposition \ref{p.induct1}.
The remaining integrals are finite by \eref{ib133}.
This proves $\B_k$ holds.
\end{proof}

\begin{lemma}  \label{lemind3}
If in Theorem \ref{MainTh}  the inequalities \eref{An1}, \eref{BnL6}, \eref{Bn1}, \eref{AnL6}
hold  for $1 \le n<k$ then
\begin{align*}
\int_0^t
s^{2k +\frac12} Y_k(s) ds \leq \bar C_{k3}(t)
\end{align*}
for some standard dominating function $\bar C_{k3}$ with $Y_k$  defined by \eref{di63}.
\end{lemma}
\begin{proof}
In view of the definition \eref{di63} of $Y_k$ we need to show that there is a standard
 dominating function $\bar C_{k3}$ such that
\begin{align}
\int_0^t
 s^{2k+ \frac12} \Big(\|Q_{k+1}(s) \|_2^2 + \kappa^2 \|P_{k+1}(s)\|_{6/5}^2
                  + \| S_k(s)\|_2^2\Big) ds \leq \bar C_{k3}(t).            \label{ind151}
\end{align}
By the bounds \eref{ib102}, \eref{ib101}, \eref{ib107} it suffices to show that each of the
following  integrals is bounded  by a standard dominating function.
\begin{align}
\int_0^t  s^{2k+ \frac12} \| A^{(i)}(s)\|_6^2 \| B^{(k-i)}(s)\|_2\| B^{(k-i)}(s)\|_6 ds, \ \ \
                                                                         1 \le i \le k                           \label{ind153}\\
\int_0^t s^{2k+ \frac12}  \| A^{(i)}(s)\|_6  \| A^{(i)}(s)\|_2\|A^{(k+1-i)}(s) \|_2^2ds, \ \ \
                                                                        \ 1\le i \le k                          \label{ind154}\\
\int_0^t s^{2k+ \frac12}\| \, [A^{(i)}(s) \wedge B^{(k-i)}(s)\|_2^2  ds,\ \ \ \   1\le i \le k  \label{ind155}\\
\int_0^t s^{2k+ \frac12}\| \, [A^{(i)}(s) \wedge P_{k-i}(s)]\, \|_2^2  ds. \ \ \ \ 1\le i <k. \label{ind156}
\end{align}
For \eref{ind153} observe that
\begin{align*}
s^{2k + \frac 12} &\| A^{(i)}(s)\|_6^2 \| B^{(k-i)}(s)\|_2\| B^{(k-i)}(s)\|_6 \\
& =   (s^{2i- \frac 12} \|A^{(i)}(s)\|_6^2)
     \, (s^{k-i+\frac 14}\|B^{(k-i)}(s)\|_2)                 \, (s^{k-i+\frac 34}\| B^{(k-i)}(s)\|_6)   \\
& \leq   (s^{2i-\frac 12} \|A^{(i)}(s)\|_6^2)\, \sqrt{ C_{(k-i)3}(t) \,  \  C_{(k-i+1)2}(t) }
\end{align*}
by the inductive hypothesis $\C_{k-i}$ of Theorem \ref{MainTh},  since $k-i < k$, and by $\B_{k-i+1}$, because $k-i+1 \le k$ for $i= 1,\dots,k$.
The integrability of the first factor is also assured by $\B_i$, which holds for $i\le k$ by Proposition \ref{p.induct2}.
Therefore the integral in \eref{ind153} is finite for $1\le i \le k$.

The integral in \eref{ind154} can be estimated as follows.
\begin{align*}
s^{2k+ \frac 12} &\| A^{(i)}(s)\|_6  \| A^{(i)}(s)\|_2\|A^{(k+1-i)}(s) \|_2^2         \notag \\
& = (s^{i-\frac 14}  \|  A^{(i)}(s)\|_6 ) \, (s^{i - \frac 34} \|  A^{(i)}(s)\|_2)
                      \, (s^{2k-2i+ \frac 32}\| A^{(k+1-i)}(s)  \|_2^2 )                            \notag \\
& \leq    (s^{i-\frac 14}  \|  A^{(i)}(s)\|_6 ) \, (s^{i - \frac 34} \|  A^{(i)}(s)\|_2) \,   C_{(k-i+1)1}(t)
\end{align*}
by $\A_{k-i+1}$, which holds for $i = 1,\dots, k$ by the hypotheses of this lemma
and    Proposition \ref{p.induct1}.
Therefore, by H\"older's inequality for the time integral, we have
\begin{align*}
\int_0^t s^{2k+ \frac12}  &\| A^{(i)}(s)\|_6  \| A^{(i)}(s)\|_2\|A^{(k+1-i)}(s) \|_2^2ds \\
&\le  \left(\int_0^t    s^{2i-\frac 12}  \|  A^{(i)}(s)\|_6^2\, ds\right)^{\frac 12}
             \left(\int_0^t  s^{2i - \frac 32} \|  A^{(i)}(s)\|_2^2\, ds \right)^{\frac 12} \,   C_{(k-i+1)1}(t) \\
& \leq  \sqrt{ C_{i2}(t)  \, C_{(i-1)3}(t) } \, C_{(k-i+1)1}(t)
\end{align*}
wherein the first integral is dominated by $\B_i$ of  Theorem \ref{MainTh},  which is valid for all $i\le k$
by the  hypotheses of this lemma and Proposition \ref{p.induct2},
 and the second integral is dominated
in accordance with $\C_{i-1}$, which is valid for $i\le k$ because $i-1 < k$.
Hence  the integral in \eref{ind154} is bounded by a standard dominating function.

The integral in \eref{ind155} can be treated exactly  as the integral in \eref{ind153},
since our use of H\"older's inequality in deriving \eref{ib102} applies equally well here.

To estimate the integral in \eref{ind156} replace $n$ by $k-i$ in the definition \eref{A11} to find
\begin{align*}
P_{k-i}(s) = \sum_{j=1}^{k-i-1} c_{(k-i)j}   [A^{(j)}(s)\wedge A^{(k-i-j)}(s) ]
\end{align*}
From this we see that it suffices to show that
\begin{align*}
\int_0^t s^{2k+\frac12} \|\, [A^{(i)}(s) \wedge [A^{(j)}(s)\wedge A^{(k-i-j)}(s) ]\, ]\,   \|_2^2 ds \leq \tilde C_{k3}(t)
\end{align*}
for some standard dominating function $\tilde C_{k3}$ for $ 1 \le i < k$
and $1\le j \le k-i-1$. But, by H\"older's inequality,
\begin{equation*}
\begin{split}
s^{2k+ \frac 12} &\|\, |A^{(i)}(s)|\, |A^{(j)}(s)|\,|A^{(k-i-j)}(s)| \,\|_2^2 \\
&\leq  \Big(s^{2i-\frac12} \|A^{(i)}(s) \|_6^2\Big) \, \Big(s^{2j+\frac 12}  \| A^{(j)}(s) \|_6^2 \Big)
\Big(s^{2k-2i-2j+\frac 12}  \| A^{(k-i-j)}(s) \|_6^2 \Big)\\
&\leq  \Big(s^{2i-\frac 12} \|A^{(i)}(s) \|_6^2\Big) \,   C_{j4}(t)  C_{(k-i-j)4}(t)
\end{split}
\end{equation*}
by $\D_j$ and $\D_{k-i-j}$, both of which hold  in accordance with the hypotheses of this lemma,
since both subscripts are strictly less than $ k$.
The integrability of the first factor  follows from $\B_i$, which holds because $ i <k$.

This completes the proof of Lemma \ref{lemind3}.
\end{proof}

\begin{proposition}[$\C_k$ holds]
 \label{p.induct3}
 Let $A(t)$ as in Proposition \ref{p.induct1}.
   If in Theorem \ref{MainTh},    \eref{An1}, \eref{BnL6}, \eref{Bn1}, \eref{AnL6}
 hold  for $1 \le n <k$   then $\C_k$ holds.
\end{proposition}
 \begin{proof}
 Setting $n = k$ in \eref{Bk2}, we see that it suffices to show that
\begin{align*}
\Big\{(2k +\frac12) \int_0^t
     s^{2k -\frac12} \|B^{(k)}(s) \|_2^2 ds
         + \int_0^t
          s^{2k +\frac 12} Y_k(s) ds\Big\} e^{\psi(t)}  \leq C_{k3}(t)
  \end{align*}
  for some standard dominating function $C_{k3}$.
By Lemma \ref{lemind3} the second integral is bounded by $\bar C_{k3}(t)$.
The first integral is also bounded by a standard bounding function
since $\A_k$ holds,  as was proven in Proposition \ref{p.induct1}.
This proves that $\C_k$ holds in view of \eref{FA2}.
\end{proof}

\begin{proposition}[$\D_k$ holds] \label{p.induct4}
  Let $A(t)$ as in Proposition \ref{p.induct1}.
    If \eref{An1}, \eref{BnL6}, \eref{Bn1}, \eref{AnL6}
  hold for  $1 \le n <k$   then $\D_k$ holds.
\end{proposition}
\begin{proof}
From \eref{di35b} and \eref{di35d} we find
\begin{align}
&\kappa^{-2} t^{2k+\frac12}\|A^{(k)}(t)\|_6^2  +\kappa^{-2}\int_0^ts^{2k+1}\|B^{(k)}(s)\|_6^2 ds   \notag\\
&\le   t^{2k+\frac12} \Big( \lambda(t) \|A^{(k)}(t)\|_2^2 + 2 \|B^{(k)}(t) \|_2^2
                               + \| R_k(t) \|_2^2 + 2\|P_k(t)\|_2^2\Big)  \label{di35b'}\\
                               &+\int_0^t s^{2k + \frac12}\Big(\lambda(s) \|B^{(k)}(s)\|_2^2  + 2 \|A^{(k+1)}(s)\|_2^2
                      + \| S_k(s)\|_2^2 + 2\|Q_{k+1} (s)\|_2^2 \Big) ds   \label{di35d'}
\end{align}
In order to prove $\D_k$  we need to show that this sum is bounded by a standard bounding
 function. Concerning the line \eref{di35b'}, the identity
\begin{align*}
 t^{2k+ \frac12} \lambda(t) \| A^{(k)}(t)\|_2^2
=\Big(t\lambda(t)\Big) \Big(  t^{2k -\frac12}  \| A^{(k)}(t)\|_2^2\Big),
\end{align*}
together with \eref{FA2t} and the already established bound $\A_k$ show this term is bounded by a standard
dominating function.
Moreover,
$$
t^{2k+ \frac12} \|B^{(k)}(t) \|_2^2 \leq C_{k3}(t)
$$
by $\C_k$,  which has already been proven in Proposition \ref{p.induct3}.
Further, $ t^{2k+ \frac12} ( \| R_k(t) \|_2^2 + 2\|P_k(t)\|_2^2)$ is suitably dominated, as has been shown in
\eref{ib137}. Thus the line \eref{di35b'} is bounded by a standard dominating function.

With respect to the line \eref{di35d'} observe that
\begin{align*}
\int_0^t s^{2k +\frac12} \lambda(s) \|B^{(k)}(s)\|_2^2 ds
&=\int_0^t \Big( s\lambda(s)\Big)  \Big( s^{2k -\frac12} \|B^{(k)}(s)\|_2^2\Big) ds \\
&\leq t\lambda(t) \, C_{k1}(t)
\end{align*}
by $\A_k$, which has been proved in  Proposition \ref{p.induct1}.
Furthermore
\[
\int_0^t  s^{2k +\frac12}  \|A^{(k+1)}(s)\|_2^2 ds \leq  C_{k3}(t)
\]
by $\C_k$, which has been proved in Proposition \ref{p.induct3}.
Thus in view of \eref{di35d'} we need only show that
\[
\int_0^t  s^{2k +\frac12} \Big( \| S_k(s)\|_2^2 + 2\|Q_{k+1} (s)\|_2^2\Big) ds \leq \bar C_{k4}(t)
\]
for some standard dominating function $\bar C_{k4}$. But this has already been shown in \eref{ind151}.
This concludes the proof of Proposition \ref{p.induct4}.
\end{proof}

\bigskip
\noindent
\begin{proof}[Proof of Theorem \ref{MainTh}]
    All of the inequalities  \eref{An1} -\eref{AnL6}
    have been established by induction under the
    assumption that the solution $A(\cdot)$  has finite action and under the technical assumption
    that the solution is smooth. The first assumption is necessary  because the bounds
    are given in terms of the action $\rho(t)$. The second assumption is needed to
     justify the computations. Here the additional hypothesis that $\|A_0\|_{H_{1/2}}$ is
      small enters because it ensures, as in Theorem \ref{thmFA}, that there is a gauge function
      $g_0 \in \G_{3/2}$ which transforms the solution to a smooth solution.
      Having such a gauge function enables the following argument. Let $A(\cdot)$ denote the
      finite action solution specified in Theorem \ref{MainTh} and let
      $\hat A(t) = A(t)^{g_0} \equiv g_0^{-1} A(t) g_0 + g_0^{-1} dg_0$ be the smooth solution
      obtained, as in Theorem \ref{thmFA}  and satisfying either \eref{N2} or \eref{D1}
      when $M \ne \R^3$.
      Since $g_0$ is time independent we have
      $(d/dt)^n \hat A(t) = g_0^{-1}  A^{(n)}(t) g_0$ for $n \ge 1$ (but not for $n =0$).
      Similarly, $(d/dt)^n \hat B(t) = g_0^{-1} B^{(n)} g_0$ for  $ n \ge 0$. Hence
      $\|(d/dt)^n \hat A(t)\|_2 = \| A^{(n)}(t)\|_2$ and $\|(d/dt)^n \hat B(t)\|_2 = \| B^{(n)}(t)\|_2$.
      Moreover $\p_j^{\hat A(t)} (g^{-1} \w g) = g^{-1} (\p_j^{A(t)} \w) g$ for any $\kf$ valued
      p-form $\w$ on $M$. Taking e.g. $\w = A^{(n)}(t)$,  this shows that
      $\| \p_j^{\hat A(t)} (d/dt)^n\hat A(t)\|_2 =   \|\p_j^{A(t)} A^{(n)}(t)\|_2$ and in particular
      $\| (d/dt)^n \hat A(t) \|_{H_1^{\hat A(t)}} = \| A^{(n)}(t)\|_{H_1^{A(t)}}$.
      (In Notation  \ref{ginvsob} we have suppressed $t$ in the subscripts.)
      In this way all of the quantities estimated in Theorem \ref{MainTh} and
       Corollary \ref{Maincor} can be estimated instead for the same gauge invariant
       functionals  of  the smooth solution $\hat A(\cdot)$. Since all of the dominating functions
       $C_{nj}$  are also gauge invariant, the inequalities of Theorem \ref{MainTh}
       and Corollary \ref{Maincor}, having been established for $\hat A$ apply equally to $A$.
This completes the proof of Theorem \ref{MainTh}.
\end{proof}

\bigskip
\noindent
\begin{proof}[Proof of Corollary \ref{Maincor}]
In the proofs of the inequalities \eref{BnL6} and \eref{AnL6}
 of Theorem \ref{MainTh}
we used the bounds \eref{di35a}, \eref{di35b}, \eref{di35c}
 and \eref{di35d} to bound the $L^6$ norms of $A^{(n)}(t)$ and $B^{(n)}(t)$.
 But the same right hand sides also bound the $H_1^A$ norms of these quantities.
 Thus if in \eref{di45d} one replaces $\kappa^{-2}\|B^{(k-1)}(t)\|_6^2$ by
 $(1/2) \|B^{(k-1)}(t)\|_{H_1^A}^2$ and one replaces in \eref{di45b} $\|A^{(k)}(s)\|_6^2$
 by  $(\kappa^2/2)\|A^{(k)}(s)\|_{H_1^A}^2$ then the proof of Proposition \ref{p.induct2}
 proves  that  the inequality \eref{BnH} of Corollary \ref{Maincor}
 holds for $n =k$. Similarly,  one need only replace the $L^6$ norms on the left hand side
 of \eref{di35b'} plus  \eref{di35d'} by  $H_1^A$ norms to find correct inequalities
 which yield the inequality \eref{AnH}
 of Corollary \ref{Maincor} with $n =k$,   via the proof of Proposition \ref{p.induct4}.
 No further induction  is needed  because the $L^2$ and $L^6$ bounds needed
  in these two  proofs have already been proven in Theorem \ref{MainTh}.
 \end{proof}

\bigskip

\begin{remark}\label{remptwise}{\rm  (Pointwise bounds) In \cite{CG2}  we derived pointwise bounds on
$A'(t,x)$ and $B(t,x)$ by a Neumann domination technique in the case $A(0)$ was in $H_1(M)$.
In that instance we took $M$ to be a compact three manifold with convex boundary.
 Pointwise bounds for $B(t,x)$ were derived  in \cite{G70}  in the case $A(0)$ is in $H_{1/2}(M)$
 and $M$ is either all of $\R^3$ or is a bounded convex set in $\R^3$ with smooth boundary.
 It seems likely that these techniques could yield pointwise bounds on all of the derivatives
 $A^{(n)}(t,x)$ and $B^{(n)}(t, x)$ with the help of the results in this paper if $M=\R^3$.
 We have not pursued this. But if $M \ne \R^3$ then some steps in the Neumann domination
  technique break down  because of boundary value problems for derivatives of $B$.
  For example if one  wishes to obtain pointwise bounds on $B'(t,x)$ when the solution
  $A(\cdot)$ satisfies  Dirichlet
  boundary conditions then the technique requires that $(d_A^*B')_{tan} =0$.
  But this boundary condition need not hold  when the solution $A(\cdot)$ merely satisfies
  Dirichlet boundary conditions.  Moreover failure to obtain the behavior of $\|B'(t)\|_{L^\infty(M)}$
  as $ t \downarrow 0$ leads, in turn, to failure to obtain pointwise bounds on $A''$,
  even though the required boundary conditions hold for $A''$.
}
\end{remark}

\bibliographystyle{amsplain}
\bibliography{ymh}

\def\polhk#1{\setbox0=\hbox{#1}{\ooalign{\hidewidth
  \lower1.5ex\hbox{`}\hidewidth\crcr\unhbox0}}}
\providecommand{\bysame}{\leavevmode\hbox to3em{\hrulefill}\thinspace}
\providecommand{\MR}{\relax\ifhmode\unskip\space\fi MR }
\providecommand{\MRhref}[2]{%
  \href{http://www.ams.org/mathscinet-getitem?mr=#1}{#2}
}
\providecommand{\href}[2]{#2}
\begin{thebibliography}{10}

\bibitem{CG1}
Nelia Charalambous and Leonard Gross, \emph{The {Y}ang-{M}ills heat semigroup
  on three-manifolds with boundary}, Comm. Math. Phys. \textbf{317} (2013),
  no.~3, 727--785. \MR{3009723}

\bibitem{CG2}
\bysame, \emph{Neumann domination for the {Y}ang-{M}ills heat equation}, J.
  Math. Phys. \textbf{56} (2015), no.~7, 073505, 21. \MR{3405967}

\bibitem{Co}
P.~E. Conner, \emph{The {N}eumann's problem for differential forms on
  {R}iemannian manifolds}, Mem. Amer. Math. Soc. \textbf{1956} (1956), no.~20,
  56. \MR{MR0078467 (17,1197e)}

\bibitem{Fuj}
Daisuke Fujiwara, \emph{Concrete characterization of the domains of fractional
  powers of some elliptic differential operators of the second order}, Proc.
  Japan Acad. \textbf{43} (1967), 82--86. \MR{0216336 (35 \#7170)}

\bibitem{GT}
David Gilbarg and Neil~S. Trudinger, \emph{Elliptic partial differential
  equations of second order}, Classics in Mathematics, Springer-Verlag, Berlin,
  2001, Reprint of the 1998 edition. \MR{MR1814364 (2001k:35004)}

\bibitem{G70}
Leonard Gross, \emph{The {Y}ang-{M}ills heat equation with finite action},
  (2016), 137 pages, http://arxiv.org/abs/1606.04151.

\bibitem{G72}
\bysame, \emph{The configuration space for {Y}ang-{M}ills fields}, In
  preparation, (2016c), 91 pages.

\bibitem{Ma3}
Antonella Marini, \emph{Dirichlet and {N}eumann boundary value problems for
  {Y}ang-{M}ills connections}, Comm. Pure Appl. Math. \textbf{45} (1992),
  no.~8, 1015--1050. \MR{MR1168118 (93k:58059)}

\bibitem{Ma4}
\bysame, \emph{Elliptic boundary value problems for connections: a non-linear
  {H}odge theory}, Mat. Contemp. \textbf{2} (1992), 195--205, Workshop on the
  Geometry and Topology of Gauge Fields (Campinas, 1991). \MR{MR1303162
  (95k:58162)}

\bibitem{Ma7}
\bysame, \emph{The generalized {N}eumann problem for {Y}ang-{M}ills
  connections}, Comm. Partial Differential Equations \textbf{24} (1999),
  no.~3-4, 665--681. \MR{MR1683053 (2000c:58025)}

\bibitem{Ma8}
\bysame, \emph{Regularity theory for the generalized {N}eumann problem for
  {Y}ang-{M}ills connections---non-trivial examples in dimensions 3 and 4},
  Math. Ann. \textbf{317} (2000), no.~1, 173--193. \MR{MR1760673 (2001i:58020)}

\bibitem{RaS}
D.~B. Ray and I.~M. Singer, \emph{{$R$}-torsion and the {L}aplacian on
  {R}iemannian manifolds}, Advances in Math. \textbf{7} (1971), 145--210.
  \MR{MR0295381 (45 \#4447)}

\end{thebibliography}

\Addresses

\end{document}